\documentclass[
oneside,
reqno,
10pt
]
{amsart}
\usepackage[utf8]{inputenc}
\usepackage{amsmath}
\usepackage{amsfonts}
\usepackage{amssymb}
\usepackage{amsthm}
\usepackage[a4paper]{geometry}
\usepackage{color}
\usepackage{latexsym}
\usepackage[colorlinks=true,linkcolor=black,citecolor=mauve]{hyperref}

\definecolor{mauve}{rgb}{0.58,0,0.82}

\newcommand{\supp}{\mathrm{supp}\,}

\newcommand{\of}{\overline{f}}

\def\d{\;\mathrm{d}}

\let\Re\relax
\DeclareMathOperator{\Re}{Re}

\def\R{\mathbb R}

\def\d{\mathrm d}

\usepackage{esint}

\let\phi\varphi
\let\epsilon\varepsilon
\newtheorem{thm}{Theorem}[section]

\newtheorem{lem}[thm]{Lemma}
\newtheorem{cor}[thm]{Corollary}
\theoremstyle{definition}
\newtheorem{df}[thm]{Definition}

\newtheorem*{rem}{Remark}

\numberwithin{equation}{section}

\date{}
\address{Simon Barth, Institut f\"ur Analysis, Dynamik und Modellierung, Universit\"at Stuttgart, Pfaffenwaldring 57, D-70569 Stuttgart, Germany}
\email{simon.barth@mathematik.uni-stuttgart.de}

\address{Andreas Bitter, Institut f\"ur Analysis, Dynamik und Modellierung, Universit\"at Stuttgart, Pfaffenwaldring 57, D-70569 Stuttgart, Germany}
\email{andreas.bitter@mathematik.uni-stuttgart.de}

\address{Semjon Wugalter, Institute for Analysis, Karlsruhe Institute of Technology (KIT), Englerstrasse 2, 76131 Karlsruhe, Germany}
\email{semjon.wugalter@kit.edu}
\author{Simon Barth, Andreas Bitter and Semjon Vugalter}
\date{}
\title[]{\ \textbf{Decay properties of zero-energy resonances 
\\of multi-particle Schr\"odinger operators 
\\and why the Efimov effect does not 
\\exist for systems of $N\geq 4$ particles} \ \\ }
\begin{document}
\maketitle 
\begin{abstract} 
We consider $N$-body Schr\"odinger operators with a virtual level at the threshold of the essential spectrum. We show that in the case of $N\geq 3$ particles in dimension $n\geq3$ virtual levels correspond to eigenvalues of the system and we obtain decay rates of the corresponding eigenfunctions in dependence on the dimension and the number of particles. We prove that in dimension $n\geq 3$ the Hamiltonian of $N\geq 4$ particles interacting via short-range potentials admits only a finite number of negative eigenvalues. We extend our results to dimension $n=1$ and $n=2$ in case of $N\geq 4$ fermions.
\end{abstract}
\pagestyle{headings}
\renewcommand{\sectionmark}[1]{\markboth{#1}{}}
\section{Introduction}
A remarkable physical phenomenon in three-body quantum systems is the so-called Efimov effect, which was first discovered by the physicist V. Efimov in 1970 \cite{Efimov}. It reads as follows: The three-body Schr\"odinger operator of three-dimensional particles interacting via short-range potentials has an infinite number of negative eigenvalues if every two-body subsystem has non-negative spectrum and at least two of them have a resonance at zero. As it was predicted by V. Efimov these three-body bound states should have very unusual properties. In particular, they accumulate logarithmically at zero with accumulation rate depending on the masses of the particles but not on the shapes of the potentials.

It became an outstanding challenge to understand this phenomenon, both from the physical and the mathematical point of view. The first mathematical proof of the Efimov effect was given by D. R. Yafaev  in \cite{Jafaev}, where he studied a symmetrized form of the Faddeev equations for the eigenvalues of the three-particle Schr\"odinger operator together with the low-energy behaviour of the resolvents of the two-body Hamiltonians. This proof constituted a major step forward in the understanding of this problem. Later he also proved that such an effect cannot occur if at least two of the two-body Hamiltonians do not have any resonances \cite{Jafaev2}. By the middle of 1990's a large number of physical and mathematical results were obtained on this topic, e.g. \cite{Sobolev, Tamura, Tamura2, Sigal, Liberating, Sem4, Semjon2, Semjon1, mag, let96}. 

A new wave of interest for the Efimov effect came at the beginning of the $21^{\mathrm{st}}$ century with the experimental discovery of this effect in an ultracold gas of caesium atoms \cite{Grimm} (for a detailed review of experimental works see \cite{Physarticle}). In 2013 the physicists Y. Nishida, S. Moroz and D. T. Son discovered the so-called super Efimov effect \cite{superefimov}, which states that in the case of three spinless fermions in dimension two the system has infinitely many negative bound states, provided every two-body subsystem admits a $p$-wave resonance at zero. Later this was proved by D. K. Gridnev \cite{Gridnev1}, applying techniques similar to \cite{Jafaev} and \cite{Sobolev}.

It is a fundamental question to ask whether the Efimov effect can be extended to multi-particle systems with more than three particles. In \cite{Liberating} Y. Nishida and S. Tan predicted that universal effects similar to the Efimov effect can be found in several types of $N$-particle systems with $N\geq 4$ in different dimensions. In 2017, Y. Nishida also predicted that a similar effect is possible in case of four two-dimensional bosons \cite{semisuper}. Here, the three-body resonances should lead to the infiniteness of the discrete spectrum of the four-body Hamiltonian. On the other hand, already in 1973 the physicists R. D. Amado and F. C. Greenwood \cite{Amado} claimed that in the case of $N\geq 4$ bosons in dimension three the Efimov effect cannot emerge if only $(N-1)$-particle subsystems have resonances. The justification of this statement in \cite{Amado} used several assumptions, which are difficult to verify. 

It was known since the 1980's that decay of solutions of the Schr\"odinger equation corresponding to virtual levels plays a crucial role for the existence of the Efimov effect. The fact that zero-energy eigenfunctions of the subsystems do not produce the Efimov effect was first proved by G. Zhislin and one of the authors of this paper in \cite{Semjon2}, where three-particle systems with two-body virtual levels were studied on spaces of states with fixed symmetries. Due to symmetry restrictions two-particle virtual levels in \cite{Semjon2} are eigenfunctions and not resonances. 

For one-particle Schr\"odinger operators with short-range potentials in dimension three solutions of the equation corresponding to virtual levels decay as $|x|^{-1}$ \cite{Jaff}, i.e. with the same decay rate as the fundamental solution of the Laplace operator in this dimension. For a subsystem with $N\geq 3$ three-dimensional particles the dimension of the corresponding space of relative motion of the particles is $3\cdot(N-1)$. The fundamental solution of the Laplace operator in this space decays as $|x|^{-(3N-5)}$, which is sufficient for a virtual level to be an eigenvalue and not a resonance for any $N\geq 3$. Due to this heuristic argument combined with \cite{Semjon2} it was always expected that $N-$particle virtual levels with $N\geq3$ in dimensions $n\geq 3$ can not produce the Efimov effect. However, to implement this argument is a very hard problem, because the sums of the potentials do not decay in all directions. Even if each of the potentials is compactly supported as a function of the distance between the particles, the sum of the potentials can not be neglected at infinity. 

The first proof that $N-$particle virtual levels for $N\geq3$ are eigenvalues of the Schr\"odinger operator was given in 2012 by D. K. Gridnev in \cite{Gridnev3} and \cite{Gridnev4}. Firstly, it was proved for $N=3$ \cite{Gridnev3}, assuming that the pair interactions $V_{ij}$ are non-positive. Later, this result was generalized to the case of $N\geq 4$ particles and it was allowed the potentials $V_{ij}$ to change signs \cite{Gridnev4}. However, some strong restrictions on the potentials, such as $V_{ij}\in L^1(\mathbb{R}^3) \cap L^3(\mathbb{R}^3)$ are required in \cite{Gridnev4} also. The method of the proof in \cite{Gridnev3, Gridnev4} is based on the analysis of the integral equation for the solution of the Schr\"odinger equation, corresponding to the virtual levels. In \cite{Gridnev2} the results of \cite{Gridnev3, Gridnev4} were applied to prove the absence of the Efimov effect in $N-$particle systems with $N\geq 4$.

In the work at hand we present a different and very transparent approach to the study of decay properties of zero-energy resonances and eigenfunctions of multi-particle Schr\"odinger operators at the edge of the essential spectrum. This approach is a further development of the Agmon's method of proving the exponential decay of eigenfunctions \cite{Ag}. It allows us to obtain estimates on the decay rates of resonances and eigenfunctions at zero energy, which in many cases are close to the optimal ones. In particular, we establish connections between the rate of decay of a virtual level at zero and Hardy's constant in the corresponding space. Since our method is purely variational it allows us to work with very weak restrictions on the potentials. In addition, as it is usual for variational methods for multi-particle Schr\"odinger operators, our approach allows us to work on subspaces with fixed permutational symmetry. Combining our results on the decay of virtual levels with the ideas of \cite{Semjon2} we give a purely variational proof of the absence of the Efimov effect for $N\geq 4$ particles in all dimensions $n\geq3$. We extend this result to systems of $N\geq 4$ identical fermions on the subspace of antisymmetric functions in dimension $n=1$ and $n=2$.

The paper is organized as follows. In Section $2$, we introduce our notations and give sufficient conditions for the existence of solutions in the space $\dot{H}^1(\mathbb{R}^d), \ d\geq 3$ of the equation
\begin{equation*}
(-\Delta+V(x))\psi=0,\quad  x\in \mathbb{R}^d,
\end{equation*}
without assuming that the potential $V(x)$ decays as $|x|\rightarrow \infty$. We then prove estimates on the rate of decay of such solutions. The conditions on the potential $V(x)$ are chosen in such a way that this result can be applied to multi-particle systems. In this case $V(x)$ will be the sum of pair interactions and $d=n(N-1)$ will be the dimension of the configuration space of a system of $N$ $n-$dimensional particles. In Section $3$, we extend this result to Schr\"odinger operators considered on subspaces of states with fixed symmetries. Section $4$ is devoted to the applications of the results obtained in Section $1$ and Section $2$. In particular, in this section we prove that for $N\geq 3$ in dimension $n\geq 3$ the virtual level is an eigenfunction. We give estimates on the rate of decay of these eigenfunctions in dependence on the number of particles and the corresponding dimension. In Section $5$ we prove the absence of the Efimov effect for $N\geq4$ particles in dimension $n\geq3$. In Section $6$ we extend the results of Section $4$ and Section $5$ to the case of $N\geq 4$ one- and two-dimensional fermions. In the Appendix we prove several technical results. Some of these results were known before and are given for the convenience of the reader only.
\section{Decay properties of zero-energy solutions of the Schr\"odinger equation}
In the following we consider the Schr\"odinger operator
\begin{equation}
h=-\Delta+V 
\end{equation}
in $L^2(\mathbb{R}^d)$, where $d\geq 3$. We assume that the potential $V$ is relatively form-bounded with relative bound zero, i.e. for every $\varepsilon >0$ there exists a constant $C(\varepsilon)>0$, such that
\begin{equation}\label{eq: A1}
\langle |V|\psi,\psi\rangle \leq \varepsilon \Vert \nabla \psi \Vert^2 + C(\varepsilon)\Vert\psi\Vert^2
\end{equation}
holds for any function $\psi \in C_0^\infty(\mathbb{R}^d)$. According to the KLMN-Theorem (see \cite{reed}, p.167) assumption \eqref{eq: A1} implies that $h$ is a self-adjoint operator in $L^2(\mathbb{R}^d)$, corresponding to the quadratic form
\begin{equation}
L[\varphi]=\Vert \nabla \varphi \Vert^2+\langle V\varphi,\varphi\rangle
\end{equation}
with form domain $H^1(\mathbb{R}^d)$. For any $\varepsilon \in (0,1)$ we denote
\begin{equation}
h_{\varepsilon}=h+\varepsilon \Delta.
\end{equation}
Let $\dot{H}^1(\mathbb{R}^d)$ be the closure of $C_0^\infty(\mathbb{R}^d)$ with respect to the gradient-norm
\begin{equation}\label{semi}
\left(\int_{\mathbb{R}^d}|\nabla \varphi|^2\, \mathrm{d}x\right)^{\frac{1}{2}}.
\end{equation}
For any self-adjoint operator $A$ we denote by $\mathcal{S}(A)$, $\mathcal{S}_{\mathrm{ess}}(A)$ and $\mathcal{S}_{\mathrm{disc}}(A)$ the spectrum, the essential spectrum and the discrete spectrum of $A$, respectively. The main result of this section is the following 
\begin{thm}\label{thm: Theorem 1}
Suppose that $V$ satisfies \eqref{eq: A1}. Further, assume that
\begin{equation}\label{1: virtual level}
h\geq 0 \qquad \quad \text{and} \qquad \quad \inf \mathcal{S}\left( h_\varepsilon \right)<0
\end{equation}
holds for any $\varepsilon \in (0,1)$. If there exist constants $\alpha_0>0$, $b>0$ and $\gamma_0\in(0,1)$, such that for any function $\psi\in H^1(\mathbb{R}^d)$ with  $\supp \psi \subset \{x\in \mathbb{R}^d:\ |x|\geq b\}$ we have
\begin{equation}\label{1: assumption for agmon}
\langle h\psi,\psi \rangle -\gamma_0 \Vert \nabla \psi\Vert^2- \langle \alpha_0^2|x|^{-2}\psi,\psi \rangle \geq 0,
\end{equation}
then the following assertions hold:
\begin{enumerate}
\item[\textbf{(i)}] If $\alpha_0>1$, then zero is a simple eigenvalue of $h$ and the corresponding eigenfunction $\varphi_0$ satisfies
\begin{equation}
\nabla \left(|x|^{\alpha_0}\varphi_0\right) \in L^2(\mathbb{R}^d) \qquad \text{and} \qquad (1+|x|)^{\alpha_0-1}\varphi_0 \in L^2(\mathbb{R}^d).
\end{equation}
Moreover, there exists a constant $\delta_0>0$, such that for any function $\psi \in H^1(\mathbb{R}^d)$ with $\langle \nabla\psi,\nabla \varphi_0\rangle =0$ it holds
\begin{equation}\label{2.8}
\langle h \psi,\psi \rangle \geq \delta_0 \Vert \nabla \psi \Vert^2.
\end{equation}
\item[\textbf{(ii)}] If $\alpha_0 \in (0,1)$ and in addition
\begin{equation}\label{1: hardyassumption}
\langle |V|\psi,\psi \rangle \leq C\Vert \nabla \psi\Vert^2
\end{equation}
holds for any function $\psi \in \dot{H}^1(\mathbb{R}^d)$ and some constant $C>0$, then there exists a non-zero function $\varphi_1 \in \dot{H}^1(\mathbb{R}^d)$ satisfying
\begin{equation}\label{1111}
\Vert \nabla \varphi_1 \Vert^2+ \langle V\varphi_1,\varphi_1 \rangle = 0.
\end{equation}
Moreover, it holds
\begin{equation}
\nabla \left(|x|^{\alpha_0}\varphi_1\right) \in L^2(\mathbb{R}^d) \qquad \text{and} \qquad (1+|x|)^{\alpha_0-1}\varphi_1 \in L^2(\mathbb{R}^d).
\end{equation}
If we assume that for some $C>0$
\begin{equation}\label{1: b. simon}
\Vert V \psi \Vert^2 \leq C\left( \Vert \nabla \psi \Vert^2 + \Vert \psi \Vert^2 \right)
\end{equation}
holds for every function $\psi\in C_0^\infty(\mathbb{R}^d)$, then the solution $\varphi_1 \in \dot{H}^1(\mathbb{R}^d)$ of \eqref{1111} is unique. Moreover, there exists a constant $\delta_1>0$, such that for any function $\psi \in \dot{H}^1(\mathbb{R}^d)$ with $\langle \nabla\psi,\nabla \varphi_1\rangle =0$ it holds 
\begin{equation}\label{2.12}
\langle h \psi,\psi \rangle \geq \delta_1 \Vert \nabla \psi \Vert^2.
\end{equation}
\item[\textbf{(iii)}] If instead of \eqref{1: assumption for agmon} a stronger inequality
\begin{equation} 
\langle h\psi,\psi\rangle - \gamma_0\Vert \nabla \psi \Vert^2 -\langle \alpha_0^2|x|^{-\beta} \psi,\psi \rangle \geq 0
\end{equation}
holds for some constants $\alpha_0,\gamma_0>0$ and $\beta\in(0,2)$, then the function $\varphi_0$ in part \textbf{(i)} of the theorem satisfies
\begin{equation}
\exp\left( \alpha_0\kappa^{-1}|x|^{\kappa} \right)\varphi_0 \in L^2(\mathbb{R}^d), \qquad \text{where}\quad \kappa=1-\frac{\beta}{2}.
\end{equation}
\end{enumerate}
\end{thm}
\begin{rem}
\begin{enumerate}
\item[\textbf{(i)}] Note that assumption \eqref{1: assumption for agmon} implies that for any $0<\varepsilon<\gamma_0$ the essential spectrum of the operator $h_\varepsilon$ is non-negative. Hence, \eqref{1: virtual level} implies that for any sufficiently small $\varepsilon>0$ the operator $h_\varepsilon$ has a discrete eigenvalue.
\item[\textbf{(ii)}] We assume \eqref{1: b. simon} to be able to apply the results by M. Schechter and B. Simon \cite{schechter} on the unique continuation theorem, which allows us to prove the uniqueness of $\varphi_1$. Without this assumption the subspace of functions in $\dot{H}^1(\mathbb{R}^d)$ satisfying \eqref{1111} is at most finite-dimensional (see Lemma \ref{A3} in the Appendix).
\item[\textbf{(iii)}] As it is mentioned in remark \textbf{(i)} the operator $h_\varepsilon$ has negative eigenvalues for small $\varepsilon>0$. We should not expect that a sequence of the corresponding eigenfunctions $\varphi_\epsilon$ always converges in $L^2(\mathbb{R}^d)$ as $\varepsilon \rightarrow 0$, because we know that for one-particle Schr\"odinger operators with short-range potentials in $\mathbb{R}^3$ this is not the case. However, if we normalize the sequence $\varphi_\varepsilon$ with the norm \eqref{semi}, condition \eqref{1: assumption for agmon} will make it energetically disadvantageous for $\varphi_\varepsilon$ to leave all compact regions. This allows us to prove that the quadratic form of $h$ has a minimizer in $\dot{H}^1(\mathbb{R}^d)$.
\item[\textbf{(iv)}] Function $\varphi_1$ in part \textbf{(ii)} of the theorem is not necessarily an eigenfunction of $h$, since it may be not an element of $L^2(\mathbb{R}^d)$. In this case zero is a resonance of $h$.
\end{enumerate}
\end{rem}
In the proof of Theorem \ref{thm: Theorem 1} we will apply the following localization error estimate, which is a straightforward modification of \cite[Lemma 5.1]{Semjon2}. For the sake of completeness we will give the corresponding proof in the Appendix.
\begin{lem}\label{eq: cmp1}
For any $\varepsilon>0$ and any fixed $b>0$ one can find $\tilde{b}>b$ and functions $\chi_1,\chi_2: \mathbb{R}^d\rightarrow \mathbb{R}$ with piecewise continuous derivatives, such that
\begin{equation}
\chi_1^2+\chi_2^2 =1, \qquad \qquad \qquad \chi_1(x)=\begin{cases}
1, &|x|\leq b
\\ 0,& |x|>\tilde{b}
\end{cases}
\end{equation}
and
\begin{equation}\label{eq: equation from cmp}
|\nabla \chi_1|^2+|\nabla \chi_2|^2 \leq \varepsilon |x|^{-2}.
\end{equation}
\end{lem}
\begin{rem}
Note that by Lemma \ref{eq: cmp1} and Hardy's inequality
\begin{equation}\label{eq: remark for cmp1}
\int |\nabla \chi_i|^2 |\psi|^2 \, \mathrm{d}x \leq 4 \varepsilon  \Vert \nabla \psi \Vert^2
\end{equation}
holds for every $\psi \in \dot{H}^1(\mathbb{R}^d)$, where $ d\geq 3$ and $i=1,2$. This estimate shows that if the constant $\tilde{b}$ is chosen much larger than $b$, then the localization error can be compensated with an $\varepsilon-$part of $\Vert \nabla \psi \Vert^2$.
\end{rem}
\subsection*{Proof of statement (i) of Theorem \ref{thm: Theorem 1} }
By Lemma \ref{A3} in the Appendix there exists a sequence of eigenfunctions $\psi_n \in H^1(\mathbb{R}^d)$, corresponding to eigenvalues $E_n<0$ of the operator $h_{n^{-1}}$, i.e. it holds
\begin{equation}\label{1: eigenfunctions}
-\left( 1-n^{-1}\right)\Delta \psi_n+V\psi_n=E_n\psi_n.
\end{equation}
We normalize the sequence $(\psi_n)_{n\in \mathbb{N}}$ by $\Vert \nabla \psi_n \Vert =1$ and take a weakly convergent subsequence, also denoted by $(\psi_n)_{n\in \mathbb{N}}$, which has a weak limit $\varphi_0 \in \dot{H}^1(\mathbb{R}^d)$. Note that by the Rellich– Kondrachov theorem $(\psi_n)_{n\in \mathbb{N}}$ converges to $\varphi_0$ in $L_{\mathrm{loc}}^2(\mathbb{R}^d)$. We will prove statement \textbf{(i)} of Theorem \ref{thm: Theorem 1} successively by the following Lemmas \ref{lemma inside proof 1} - \ref{1:lem 2}.
\begin{lem}\label{lemma inside proof 1}
The weak limit $\varphi_0\in \dot{H}^1(\mathbb{R}^d)$ of the sequence $(\psi_n)_{n\in \mathbb{N}}$ is not identically zero.
\end{lem}
\begin{proof}
We consider the functional
\begin{equation}
L[\psi,\varepsilon]:= (1-\varepsilon)\Vert \nabla \psi \Vert^2+\langle V\psi,\psi \rangle,
\end{equation}
where $\psi \in H^1(\mathbb{R}^d)$ and $\varepsilon>0$. Let $b>0$, such that \eqref{1: assumption for agmon} holds. We fix $\varepsilon_1>0$ and construct functions $\chi_1$, $\chi_2$ in accordance with Lemma  \ref{eq: cmp1}, which implies
\begin{equation}\label{eq: sum chi1 and chi2}
L[\psi,\varepsilon]\geq L[\psi \chi_1,\varepsilon+\varepsilon_1] + L[\psi \chi_2,\varepsilon+\varepsilon_1]
\end{equation}
for every $\psi \in H^1(\mathbb{R}^d)$ independently of $\varepsilon$. Since the operator $h$ is non-negative we have
\begin{align}\label{eq: chi1}
\begin{split}
L[\psi\chi_1,\varepsilon+\varepsilon_1]&=(1-\varepsilon-\varepsilon_1)\Vert \nabla (\psi\chi_1) \Vert^2 + \langle V\psi\chi_1,\psi\chi_1 \rangle
\\ &\geq -(\varepsilon+\varepsilon_1)\Vert \nabla (\psi\chi_1) \Vert^2.
\end{split}
\end{align}
In addition, since $\supp (\psi\chi_2) \subset \{x\in \mathbb{R}^d : |x|\geq b\}$ we conclude by \eqref{1: assumption for agmon} that
\begin{align}\label{eq: chi2}
\begin{split}
L[\psi\chi_2,\varepsilon+\varepsilon_1]&=(1-\varepsilon-\varepsilon_1)\Vert \nabla (\psi\chi_2) \Vert^2 + \langle V\psi\chi_2,\psi\chi_2 \rangle
\\ &= (1-\gamma_0)\Vert \nabla (\psi\chi_2) \Vert^2 + \langle V\psi\chi_2,\psi\chi_2 \rangle+(\gamma_0-\varepsilon-\varepsilon_1)\Vert \nabla (\psi\chi_2) \Vert^2
\\ &\geq (\gamma_0-\varepsilon-\varepsilon_1)\Vert \nabla (\psi\chi_2) \Vert^2.
\end{split}
\end{align}
Hence, \eqref{eq: chi1} and \eqref{eq: chi2} imply
\begin{equation}\label{eq: estimate for solution not zero}
L[\psi,\varepsilon] \geq -(\varepsilon+\varepsilon_1)\Vert \nabla (\psi\chi_1) \Vert^2+(\gamma_0-\varepsilon-\varepsilon_1)\Vert \nabla (\psi\chi_2) \Vert^2.
\end{equation}
For $\psi=\psi_n$ and $\varepsilon = n^{-1}$, estimate \eqref{eq: estimate for solution not zero} yields
\begin{equation}
-(\varepsilon_1+n^{-1})\Vert \nabla (\psi_n \chi_1) \Vert^2 +(\gamma_0-\varepsilon_1-n^{-1})\Vert \nabla(\psi_n \chi_2) \Vert^2<0,
\end{equation}
which implies
\begin{equation}\label{1: eq <0}
(\gamma_0-\varepsilon_1-n^{-1})\left(\Vert \nabla (\psi_n \chi_1) \Vert^2 +\Vert \nabla (\psi_n \chi_2) \Vert^2\right)<\gamma_0\Vert \nabla (\psi_n \chi_1) \Vert^2.
\end{equation}
By the IMS localization formula we have
\begin{equation}
\Vert \nabla (\psi_n \chi_1) \Vert^2+\Vert \nabla (\psi_n \chi_2) \Vert^2 \geq \Vert \nabla \psi_n \Vert^2=1
\end{equation}
for every $n\in \mathbb{N}$. Hence, by \eqref{1: eq <0} we obtain
\begin{equation}\label{eq: estimate 1-eps2}
\Vert \nabla(\psi_n \chi_1) \Vert^2 \geq \frac{\gamma_0-\varepsilon_1-n^{-1}}{\gamma_0}\geq 1 -\varepsilon_2,
\end{equation}
where $\varepsilon_2>0$ can be chosen arbitrarily small by choosing $\varepsilon_1>0$ sufficiently small and $n\in \mathbb{N}$ sufficiently large.
Due to \eqref{eq: chi2} with $\varepsilon=n^{-1}$ we have $L[\psi_n\chi_2,n^{-1}+\varepsilon_1]>0$. This, together with \eqref{eq: sum chi1 and chi2} and $L[\psi_n,n^{-1}]<0$ implies
\begin{align}\label{eq: for not zero}
\begin{split}
0 &>L[\psi_n\chi_1,n^{-1}+\varepsilon_1] = \left(1-n^{-1}-\varepsilon_1 \right) \Vert \nabla(\psi_n \chi_1) \Vert^2 + \langle V\psi_n \chi_1,\psi_n \chi_1 \rangle
\\ &\geq \left(1-n^{-1}-2\varepsilon_1 \right) \Vert \nabla(\psi_n \chi_1) \Vert^2 -C(\varepsilon_1) \Vert \psi_n\chi_1 \Vert^2,
\end{split}
\end{align}
where in the last inequality we used \eqref{eq: A1}. Combining \eqref{eq: for not zero} and \eqref{eq: estimate 1-eps2} we arrive at
\begin{equation}
\Vert \psi_n \chi_1 \Vert^2 \geq \frac{(1-n^{-1}-2\varepsilon_1)(1-\varepsilon_2)}{C(\varepsilon_1)}.
\end{equation}
Since $\chi_1$ is compactly supported, $|\chi_1|\leq 1$ and $(\psi_n)_{n\in \mathbb{N}}$ converges to $\varphi_0$ in $L_{\mathrm{loc}}^2(\mathbb{R}^d)$, the last inequality proves the Lemma. 
\end{proof}
\begin{rem}
Since 
\begin{equation}\label{1-1}
\Vert \nabla(\chi_1 \psi_n) \Vert^2 + \Vert \nabla (\chi_2 \psi_n) \Vert^2 = \Vert \nabla\psi_n \Vert^2 + \int \left( |\nabla \chi_1|^2+|\nabla\chi_2|^2 \right) |\psi_n|^2\, \mathrm{d}x,
\end{equation}
inequality \eqref{eq: remark for cmp1} shows that the last term on the r.h.s. of \eqref{1-1} can be estimated as $\varepsilon\Vert \nabla \psi_n \Vert^2 = \varepsilon$. This implies
\begin{equation}\label{1-2}
\Vert \nabla(\chi_2\psi_n)\Vert^2 \leq (1+\varepsilon) - \Vert \nabla(\chi_1 \psi_n) \Vert^2.
\end{equation}
Combining \eqref{1-2} with \eqref{eq: estimate 1-eps2} yields $\Vert \nabla (\chi_2 \psi_n)\Vert^2 \leq \tilde{\varepsilon}$, where $\tilde{\varepsilon}>0$ can be chosen arbitrarily small for large $\tilde{b}$ and $n$. We will use this estimate in the proof of Theorem \ref{thm: Theorem 1}.
\end{rem}
\begin{lem}\label{1: lem uniformly bound of psi n}
Assume that \eqref{1: virtual level} and \eqref{1: assumption for agmon} hold for some $\alpha_0>1$. Then there exists a constant $C>0$, such that for any eigenfunction $\psi_n\in H^1(\mathbb{R}^d)$ corresponding to a negative eigenvalue of the operator $h_{n^{-1}}$, normalized by $\Vert \nabla \psi_n \Vert=1$, we have 
\begin{equation}
\Vert \nabla (|x|^{\alpha_0}\psi_n) \Vert \leq C \qquad \text{and} \qquad \Vert (1+|x|)^{\alpha_0-1}\psi_n \Vert\leq C.
\end{equation} 
\end{lem}
\begin{rem}
Recall that eigenfunctions $\psi_n$ of the operators $h_{n^{-1}}$ decay exponentially with powers depending on the distances from the corresponding eigenvalues to zero. Since for $n\rightarrow \infty$ the negative eigenvalues of $h_{n^{-1}}$ converge to zero, these exponential estimates are not uniform in $n\in \mathbb{N}$. However, Lemma \ref{1: lem uniformly bound of psi n} shows that if condition \eqref{1: assumption for agmon} holds for functions supported far from the origin, a uniform estimate on the rate of decay of eigenfunctions of $h_{n^{-1}}$ exists. This estimate is of the polynomial type and the corresponding power depends on the parameter $\alpha_0$ in \eqref{1: assumption for agmon} only.
\end{rem}
\begin{proof}[Proof of Lemma \ref{1: lem uniformly bound of psi n}]
For any $\varepsilon>0$ and $R>0$ we define the function
\begin{equation}\label{1: function agmon}
G_\varepsilon(x)=\frac{|x|^{\alpha_0}}{1+\varepsilon|x|^{\alpha_0}} \chi_R(x),
\end{equation}
where $\chi_R$ is a $C^\infty$ cutoff function, such that
\begin{equation}\label{11chi}
\chi_R(x)=\begin{cases}
0, & |x|\leq R,
\\ 1, & |x|\geq 2R.
\end{cases}
\end{equation}
Since for the eigenfunctions $\psi_n$ we have
\begin{equation}\label{1: multiply and int by parts}
-(1-n^{-1})\Delta\psi_n+V\psi_n=E_n\psi_n
\end{equation}
with $E_n<0$ and each $\psi_n$ decays exponentially, we can multiply \eqref{1: multiply and int by parts} with $G^2_\varepsilon\overline{\psi_n}$ and integrate by parts to obtain
\begin{equation}\label{eq: eigenfunctions psi_n}
\left( 1-n^{-1}\right) \langle \nabla \psi_n, \nabla\left(G_\varepsilon^2\psi_n\right) \rangle+ \langle V\psi_n , G_\varepsilon^2\psi_n \rangle= E_n\Vert G_\varepsilon \psi_n \Vert^2<0.
\end{equation}
Since 
\begin{equation}
\Re \langle V\psi_n,G_\varepsilon^2\psi_n \rangle = \langle V\psi_n,G_\varepsilon^2\psi_n \rangle \qquad \text{and} \qquad \Re E_n \Vert G_\varepsilon \psi_n \Vert^2 = E_n\Vert G_\varepsilon \psi_n \Vert^2
\end{equation}
we have
\begin{equation}
\Re \langle \nabla \psi_n,\nabla \left( G_\varepsilon^2 \psi_n \right)\rangle = \langle \nabla \psi_n,\nabla \left( G_\varepsilon^2 \psi_n \right)\rangle.
\end{equation}
Note that 
\begin{align}\label{eq: agmon Re}
\begin{split}
\Re \langle \nabla \psi_n , \nabla (G_\varepsilon^2 \psi_n) \rangle &= \Re \langle \nabla \psi_n, G_\varepsilon \psi_n \nabla G_\varepsilon \rangle  + \Re   \langle (\nabla \psi_n) G_\varepsilon , \nabla (G_\varepsilon \psi_n)\rangle
\\ &= \Re \langle \nabla(\psi_n G_\varepsilon), \psi_n \nabla G_\varepsilon \rangle - \Re \langle \psi_n\nabla G_\varepsilon, \psi_n \nabla G_\varepsilon \rangle  
\\ & \ \ \ +\Re \langle \nabla(\psi_n G_\varepsilon), \nabla (\psi_n G_\varepsilon) \rangle - \Re \langle \psi_n \nabla G_\varepsilon, \nabla (\psi_n G_\varepsilon) \rangle 
\\ &=\Re \langle \nabla(\psi_n G_\varepsilon), \nabla (\psi_n G_\varepsilon) \rangle - \Re \langle \psi_n \nabla G_\varepsilon, \psi_n \nabla G_\varepsilon \rangle.
\end{split}
\end{align}
This implies
\begin{equation}
\langle \nabla \psi_n, \nabla (G_\varepsilon^2 \psi_n) \rangle= \Vert \nabla(\psi_n G_\varepsilon)\Vert^2 - \Vert \psi_n \nabla G_\varepsilon\Vert^2,
\end{equation}
which together with \eqref{eq: eigenfunctions psi_n} yields
\begin{equation}\label{1: agmon with psin}
\left(1-\frac{1}{n} \right)\left(\Vert \nabla(\psi_n G_\varepsilon) \Vert^2 - \int|\psi_n|^2 |\nabla G_\varepsilon|^2\, \mathrm{d}x\right)+\int V|\psi_n G_\varepsilon|^2\, \mathrm{d}x<0.
\end{equation}
For $|x|>2R$ we can estimate
\begin{equation}\label{1: estimate G}
|\nabla G_\varepsilon| = \frac{\alpha_0 |x|^{\alpha_0-1}}{(1+\varepsilon|x|^{\alpha_0})^2}\leq \alpha_0|x|^{-1}|G_\varepsilon|.
\end{equation} 
For $|x| \in[R,2R]$ the function $|\nabla G_\varepsilon|$ is uniformly bounded in $\varepsilon$, which together with Hardy's inequality implies
\begin{equation}\label{1: eq with C(R)}
\int_{\{R\leq |x|\leq 2R\}} |G_\varepsilon|^2|\psi_n|^2\, \mathrm{d}x \leq C\int_{\{R\leq |x|\leq 2R\}} |\psi_n|^2\, \mathrm{d}x \leq \widetilde{C} R^2 \int |\nabla \psi_n|^2 \, \mathrm{d}x=: C_0.
\end{equation}
Substituting \eqref{1: estimate G} and \eqref{1: eq with C(R)} into \eqref{1: agmon with psin} we obtain
\begin{equation}\label{1: 39}
\left(1-n^{-1}\right)\Vert \nabla(\psi_n G_\varepsilon) \Vert^2 +\langle V G_\varepsilon \psi_n,G_\varepsilon \psi_n \rangle - \alpha_0^2 \int_{\{|x|>2R\}}\frac{|G_\varepsilon \psi_n|^2}{|x|^2}\, \mathrm{d}x\leq C_1,
\end{equation}
where $C_1>0$ does not depend on $n\in \mathbb{N}$ or $\varepsilon>0$. Note that the function $G_\varepsilon \psi_n$ is supported outside the ball with radius $R>0$. For $R>b$ it satisfies \eqref{1: assumption for agmon}, i.e. it holds
\begin{equation}\label{1gamma}
(1-\gamma_0)\Vert \nabla (G_\varepsilon \psi_n)\Vert^2+\langle VG_\varepsilon \psi_n,G_\varepsilon \psi_n \rangle-\alpha_0^2 \langle |x|^{-2}G_\varepsilon \psi_n,G_\varepsilon \psi_n \rangle \geq 0.
\end{equation}
For $n>2\gamma_0^{-1} $ estimates \eqref{1: 39} and \eqref{1gamma} imply
\begin{equation}\label{241}
\frac{\gamma_0}{2} \Vert \nabla(G_\varepsilon \psi_n) \Vert^2 \leq C_1.
\end{equation}
Taking $\varepsilon\rightarrow 0$ yields $\Vert \nabla \left( |x|^{\alpha_0}\psi_n\right) \Vert \leq C$, which together with Hardy's inequality completes the proof.
\end{proof}
\begin{lem}\label{1Lem 1} 
Assume that \eqref{1: virtual level} and \eqref{1: assumption for agmon} hold for $\alpha_0>1$. Then zero is an eigenvalue of $h$ and the corresponding eigenfunction $\varphi_0$ satisfies
\begin{equation}\label{1: decay rate}
\nabla \left(|x|^{\alpha_0}\varphi_0 \right)\in L^2(\mathbb{R}^d) \qquad \text{and} \qquad (1+|x|)^{\alpha_0-1}\varphi_0 \in L^2(\mathbb{R}^d).
\end{equation}
\end{lem}
\begin{proof}
We take a sequence of eigenfunctions $\psi_n$ of $h_{n^{-1}}$ normalized by $\Vert \nabla \psi_n \Vert = 1$. This sequence has a subsequence (also denoted by $(\psi_n)_{n\in \mathbb{N}})$ with a weak limit $\varphi_0 \in \dot{H}^1(\mathbb{R}^d)$. According to Lemma \ref{lemma inside proof 1} we have $\varphi_0 \not \equiv 0$. Since $(\psi_n)_{n\in \mathbb{N}}$ converges to $\varphi_0$ in $L_{\mathrm{loc}}^2(\mathbb{R}^d)$ and by Lemma \ref{1: lem uniformly bound of psi n} we have $\Vert (1+|x|)^{\alpha_0-1}\psi_n \Vert \leq C$ for $\alpha_0>1$ and $C$ independent of $n\in \mathbb{N}$, we conclude that $(1+|x|)^{\alpha_0-1}\varphi_0 \in L^2(\mathbb{R}^d)$ holds. Furthermore, this also shows that $\langle V\varphi_0,\varphi_0\rangle$ is well defined. Our next goal is to prove $\langle V\varphi_0,\varphi_0 \rangle =-1$. We write
\begin{align}\label{eq:1}
\langle V\varphi_0,\varphi_0\rangle &= \langle V\varphi_0, \varphi_0-\psi_n \rangle + \langle V \varphi_0,\psi_n \rangle \notag
\\ &= \langle V\varphi_0, \varphi_0-\psi_n \rangle + \langle V(\varphi_0-\psi_n),\psi_n \rangle +\langle V\psi_n,\psi_n \rangle.
\end{align}
Due to \eqref{eq: A1} the first term on the r.h.s. of \eqref{eq:1} can be estimated by
\begin{align}\label{eq:2}
|\langle V\varphi_0, \varphi_0-\psi_n \rangle| &\leq \langle |V|^\frac{1}{2} \varphi_0, |V|^\frac{1}{2} |\varphi_0-\psi_n| \rangle \notag
\\ &\leq \left( \Vert \nabla \varphi_0 \Vert^2 + C(1) \Vert \varphi_0\Vert^2\right)^{\frac{1}{2}}\left(\varepsilon \Vert \nabla(\varphi_0-\psi_n) \Vert^2 + C(\varepsilon) \Vert \varphi_0-\psi_n\Vert^2 \right)^{\frac{1}{2}} \notag
\\ &\leq C\left(2\varepsilon \left(\Vert \nabla\varphi_0 \Vert^2 +\Vert \nabla \psi_n\Vert^2\right)+ C(\varepsilon) \Vert \varphi_0-\psi_n\Vert^2 \right)^{\frac{1}{2}}.
\end{align}
Note that by the semicontinuity of the norm we have $\Vert \nabla \varphi_0 \Vert \leq 1$. Since $\Vert \psi_n-\varphi_0 \Vert \rightarrow 0$ as $n\rightarrow \infty$, choosing $\varepsilon>0$ sufficiently small and $n\in \mathbb{N}$ sufficiently large we can get the r.h.s. of \eqref{eq:2} arbitrarily small. Similar arguments show that the second term on the r.h.s. of \eqref{eq:1} can be done arbitrarily small as well. Consequently, we have $\langle V\psi_n,\psi_n\rangle \rightarrow \langle V\varphi_0,\varphi_0 \rangle$ as $n \rightarrow \infty.$ By
\begin{equation}
(1-n^{-1})\Vert \nabla\psi_n \Vert^2 +\langle V\psi_n,\psi_n \rangle \leq 0 \qquad \text{and} \qquad \Vert \nabla \psi_n \Vert=1
\end{equation}
we conclude $\langle V\varphi_0,\varphi_0\rangle = -1$. Since $\Vert\nabla \varphi_0 \Vert \leq 1$,	 we have
\begin{equation}\label{1: 240}
\Vert \nabla \varphi_0 \Vert^2+\langle V\varphi_0,\varphi_0 \rangle\leq 0.
\end{equation}
Together with $h\geq 0$ this implies $\Vert \nabla \varphi_0 \Vert=1$. Hence, $\varphi_0$ is a minimizer of the quadratic form of $h$ and an eigenfunction of $h$, corresponding to the eigenvalue zero. Finally, repeating the same arguments for $\varphi_0$ as we used in Lemma \ref{1Lem 1} to get \eqref{241} for the eigenfunctions $\psi_n$, we obtain $\nabla(|x|^{\alpha_0}\varphi_0) \in L^2(\mathbb{R}^d)$.
\end{proof}
Our next goal is to prove inequality \eqref{2.8} and the nondegeneracy of $\varphi_0$. We will do it in the following Lemmas \ref{1: lemma all eigenfunctions converge} - \ref{1:lem 2}.
\begin{lem}\label{1: lemma all eigenfunctions converge}
For any $\varepsilon>0$ one can find $n_0 \in \mathbb{N}$, such that for any $n\geq n_0$ and any eigenfunction $\psi_n$ with $\Vert \nabla \psi_n \Vert =1$, corresponding to some negative eigenvalue of the operator $h_{n^{-1}}$, it holds $\Vert \psi_n-\varphi_0 \Vert<\varepsilon$.
\end{lem}
\begin{proof}
Assume that we have a sequence of eigenfunctions $\psi_n \in H^1(\mathbb{R}^d),\ \Vert \nabla \psi_n \Vert=1$, corresponding to some negative eigenvalues of the operator $h_{n^{-1}}$ for $ n\in \mathbb{N}$. Furthermore, we assume that $\Vert \psi_{n}-\varphi_0 \Vert\geq C>0$ holds for every $n\in \mathbb{N}$. Proceeding as in the proof of Lemmas \ref{lemma inside proof 1} and \ref{1Lem 1} we can find a subsequence, also denoted by $(\psi_n)_{n\in \mathbb{N}}$, such that $(\psi_n)_{n\in \mathbb{N}}$ converges to some function $\tilde{\varphi}_0 \in H^1(\mathbb{R}^d)$ with $\tilde{\varphi}_0\not \equiv 0, \ \Vert\nabla \tilde{\varphi}_0\Vert = 1$ and
\begin{equation}
\Vert \nabla \tilde{\varphi}_0 \Vert^2+\langle V\tilde{\varphi}_0,\tilde{\varphi}_0 \rangle = 0.
\end{equation}
By $\Vert \nabla \varphi_0\Vert = \Vert\nabla \tilde{\varphi}_0\Vert = 1$ and $\Vert \psi_{n}-\varphi_0 \Vert\geq C>0$ we conclude that $\varphi_0$ and $\tilde{\varphi}_0$ are linearly independent. According to \cite{Goelden} an eigenvalue of a Schr\"odinger operator coinciding with the bottom of the spectrum cannot be degenerate. Consequently, $\varphi_0$ and $\tilde{\varphi}_0$ cannot be linearly independent.
\end{proof}
\begin{lem}\label{1: corollary}
For any sufficiently small $\varepsilon>0$ the operator $h_\varepsilon$ has only one negative eigenvalue, which is non-degenerate.
\end{lem}
\begin{proof}
Assume there is a sequence $a(n)\in(0,1)$ with $a(n)\rightarrow 0$ as $n\rightarrow \infty$, such that for any $n\in \mathbb{N}$ the operator $h_{a(n)}=-(1-a(n))\Delta+V$ has at least two eigenvalues. Recall that the lowest eigenvalue of $h_{a(n)}$ is non-degenerate. We consider two eigenfunctions $\psi_{n}^{(1)}$ and $\psi_{n}^{(2)}$ of $h_{a(n)}$, normalized by $\Vert \psi_{n}^{(1)}\Vert =\Vert \psi_{n}^{(2)}\Vert=1$, where $\psi_{n}^{(1)}$ corresponds to the lowest eigenvalue. Now $\psi_{n}^{(1)}$ and $\psi_{n}^{(2)}$ are orthogonal in $L^2(\mathbb{R}^d)$ and by Lemma \ref{1: lemma all eigenfunctions converge} $\psi_{n}^{(1)}$ and $\psi_{n}^{(2)}$ both converge to $\varphi_0 \in L^2(\mathbb{R}^d)$, which is a contradiction.
\end{proof}
\begin{lem}\label{1:lem 2}
There exists a constant $\delta_0>0$, such that for every function $\psi \in H^1(\mathbb{R}^d)$ with $\langle \nabla \psi, \nabla \varphi_0 \rangle =0$ it holds
\begin{equation}
(1-\delta_0)\Vert \nabla \psi \Vert^2 + \langle V\psi,\psi\rangle \geq 0.
\end{equation}
\end{lem}
\begin{proof}
We prove the Lemma by contradiction. Assume that there is no such constant $\delta_0>0$. Then there exists a sequence of functions $g_{n}\in H^1(\mathbb{R}^d)$ with 
\begin{equation}\label{1: gn<0}
\langle \nabla g_n, \nabla \varphi_0 \rangle =0 \qquad \text{and} \qquad \langle h_{n^{-1}}g_n,g_n \rangle <0.
\end{equation}
Note that for $c_1,c_2 \in \mathbb{C}$ we have
\begin{align}
\begin{split}
\langle h_{n^{-1}}(c_1g_n+c_2\varphi_0), (c_1g_n+c_2\varphi_0) \rangle &=
c_1^2\langle h_{n^{-1}} g_n,g_n \rangle + c_2^2 \langle h_{n^{-1}} \varphi_0,\varphi_0 \rangle 
\\ &\ \ \ \ \ \ \ \ \ \ \ \ \ \qquad \ \ \ \ \ \ \ + 2\Re c_1 \overline{c_2} \langle h_{n^{-1}}g_n,\varphi_0 \rangle.
\end{split}
\end{align}
Further, it is easy to see that
\begin{equation}
\Re\langle h_{n^{-1}} g_n,\varphi_0 \rangle = \Re\langle g_n, h \varphi_0 \rangle-n^{-1}\Re \langle \nabla g_n,\nabla \varphi_0 \rangle=0
\end{equation} and
\begin{equation}
\langle h_{n^{-1}}\varphi_0,\varphi_0\rangle = \langle h\varphi_0,\varphi_0\rangle - n^{-1} \Vert \nabla \varphi_0 \Vert^2 = -n^{-1}
\end{equation}
hold for every $n \in \mathbb{N}$. Hence, we conclude that for any linear combination $c_1 g_n+c_2\varphi_0$ we have
\begin{equation}\label{eq: linear combination negative}
\langle h_{n^{-1}}(c_1g_n+c_2\varphi_0), (c_1g_n+c_2\varphi_0) \rangle<0.
\end{equation}
Since by \eqref{1: gn<0} the functions $\varphi_0$ and $g_n$ are linearly indpendent, for any $n\in \mathbb{N}$ we can find a linear combination $f_n$ of $\varphi_0$ and $g_n$, such that $f_n$ is orthogonal to the ground state of $h_{n^{-1}}$. According to Lemma \ref{1: corollary} for sufficiently large $n\in \mathbb{N}$ the operator $h_{n^{-1}}$ has only one negative eigenvalue, which yields $\langle h_{n^{-1}}f_n,f_n \rangle\geq 0$. This is a contradiction to \eqref{eq: linear combination negative}.
\end{proof}
Combining Lemma \ref{1Lem 1} and Lemma \ref{1:lem 2} proves statement \textbf{(i)} of Theorem \ref{thm: Theorem 1}.
\subsection*{Proof of statements (ii) and (iii) of Theorem \ref{thm: Theorem 1} }
Note that for $\alpha_0\in(0,1)$ the sequence of eigenfunctions $\psi_n$ of the operators $h_{n^{-1}}$, normalized by $\Vert \nabla \psi_n \Vert=1$, does not necessarily converge in $L^2(\mathbb{R}^d)$, as for example happens in the case of a one-particle Schr\"odinger operator in $\mathbb{R}^3$. To ensure that the functional $\Vert \nabla\psi \Vert^2+\langle V\psi,\psi\rangle$ is well defined for the weak limit $\varphi_1 \in \dot{H}^1(\mathbb{R}^d)$ and that $\langle V\psi_n,\psi_n\rangle$ converges to $\langle V\varphi_1,\varphi_1 \rangle$ as $n \rightarrow \infty$, we assume \eqref{1: hardyassumption}. We will prove part $\textbf{(ii)}$ of Theorem \ref{thm: Theorem 1} in two steps. In Lemma \ref{1Lem 3} we prove the existence of a function $\varphi_1$ satisfying \eqref{1111}. Then, in Lemma \ref{210} we prove the uniqueness of $\varphi_1$ and the inequality \eqref{2.12}.
\begin{lem}\label{1Lem 3}
Assume that \eqref{1: virtual level} and \eqref{1: assumption for agmon} hold for $\alpha_0\in(0,1)$ and in addition
\begin{equation}\label{1: hardy assumption}
\langle |V|\psi,\psi \rangle \leq C\Vert \nabla \psi\Vert^2
\end{equation}
holds for any function $\psi \in \dot{H}^1(\mathbb{R}^d)$ and some constant $C>0$. Then, there exists a function $\varphi_1 \in \dot{H}^1(\mathbb{R}^d)$ with
\begin{equation}
\Vert \nabla \varphi_1 \Vert^2+ \langle V\varphi_1,\varphi_1 \rangle = 0.
\end{equation}
Moreover, $\varphi_1$ satisfies 
\begin{equation}
\nabla \left(|x|^{\alpha_0} \varphi_1 \right)\in L^2(\mathbb{R}^d) \qquad \text{and} \qquad (1+|x|)^{\alpha_0-1}\varphi_1 \in L^2(\mathbb{R}^d).
\end{equation}
\end{lem}
\begin{proof}
By assumption \eqref{1: virtual level} there exists a sequence of functions $\psi_n \in \dot{H}^1(\mathbb{R}^d)$ satisfying
\begin{equation}
\left(1-n^{-1}\right)\Vert \nabla \psi_n \Vert^2 +\langle V\psi_n,\psi_n \rangle<0 \qquad \text{and} \qquad \Vert \nabla \psi_n \Vert=1.
\end{equation}
Repeating the same arguments as in Lemma \ref{lemma inside proof 1} shows that there is a subsequence, also denoted by $(\psi_n)_{n\in \mathbb{N}}$, which converges in $L_{\mathrm{loc}}^2(\mathbb{R}^d)$ to some function $\varphi_1\in \dot{H}^1(\mathbb{R}^d)$. Let us prove that $\varphi_1$ is a minimizer of the quadratic form of $h$ in $\dot{H}^1(\mathbb{R}^d)$ by showing that $\langle V\varphi_1,\varphi_1 \rangle =-1$. We fix the constant $b>0$ and construct functions $\chi_1,\chi_2$ according to Lemma \ref{eq: cmp1}. Since $\chi_1^2+\chi_2^2=1$ we have
\begin{equation}
\langle V\varphi_1,\varphi_1 \rangle = \langle V\varphi_1,\varphi_1\chi_1^2 \rangle + \langle V\varphi_1,\varphi_1\chi_2^2 \rangle.
\end{equation}
Note that
\begin{align}
\langle V\varphi_1,\varphi_1\chi_1^2 \rangle &= \langle V(\varphi_1-\psi_n),\varphi_1\chi_1^2 \rangle + \langle V \psi_n ,\varphi_1 \chi_1^2 \rangle \notag
\\ &=\langle V(\varphi_1-\psi_n),\varphi_1\chi_1^2 \rangle+\langle V \psi_n,\psi_n\chi_1^2 \rangle +\langle V\psi_n,(\varphi_1-\psi_n)\chi_1^2 \rangle.
 \label{eq: estimate with three summands}
\end{align} 
At first we estimate the first term on the r.h.s. of \eqref{eq: estimate with three summands}. It holds
\begin{align}
|\langle V(\varphi_1-\psi_n),\varphi_1\chi_1^2  \rangle| &\leq \Vert |V|^{\frac{1}{2}}\chi_1(\varphi_1-\psi_n) \Vert \cdot \Vert |V|^{\frac{1}{2}} \varphi_1 \Vert \leq C\Vert |V|^{\frac{1}{2}}\chi_1(\varphi_1-\psi_n) \Vert\cdot \Vert \nabla \varphi_1\Vert \notag
\\ &\leq C \left(\varepsilon \Vert \nabla\left( \chi_1(\varphi_1-\psi_n)\right) \Vert^2  + C(\varepsilon)\Vert \chi_1(\varphi_1-\psi_n) \Vert^2\right)^{\frac{1}{2}}.
\end{align}
Here we used \eqref{eq: A1}, \eqref{1: hardyassumption}, $|\chi_1|\leq 1$ and $\Vert \nabla \varphi_1 \Vert \leq 1$. Moreover, it holds
\begin{equation}
\Vert \nabla(\chi_1(\varphi_1-\psi_n)) \Vert^2 \leq 2\Vert \nabla\chi_1 \Vert^2 \Vert \varphi_1-\psi_n \Vert_{\supp(\chi_1)}^2 + 2\Vert \nabla(\varphi_1-\psi_n) \Vert^2.
\end{equation}
Since $\psi_n \rightarrow \varphi_1$ in $L_{\mathrm{loc}}^2(\mathbb{R}^d)$ and $\chi_1$ is compactly supported, for fixed $\varepsilon_1>0$ and large $n\in \mathbb{N}$ we get
\begin{equation}
\Vert \nabla(\chi_1(\varphi_1-\psi_n)) \Vert^2 \leq 2\varepsilon_1+4\Vert \nabla \varphi_1 \Vert^2 +4 \Vert \nabla \psi_n \Vert^2 \leq 9.
\end{equation}
For any fixed $\tilde{\varepsilon}>0$ and large $n$ this implies
\begin{equation}
|\langle V(\varphi_1-\psi_n),\chi_1^2\varphi_1 \rangle | \leq C \left( 9\varepsilon+C(\varepsilon)\Vert \chi_1(\varphi_1-\psi_n) \Vert^2 \right)^{\frac{1}{2}} \leq \tilde{\varepsilon}.
\end{equation}
Applying similar arguments to the last term on the r.h.s. of \eqref{eq: estimate with three summands} yields
\begin{equation}
|\langle V\psi_n\chi_1,(\varphi_1-\psi_n)\chi_1 \rangle| \leq \tilde{\varepsilon}.
\end{equation}
Hence, it holds
\begin{equation}\label{eq: estimate for final 1}
\langle V \varphi_1\chi_1,\varphi_1\chi_1 \rangle \leq \langle V\psi_n \chi_1,\psi_n \chi_1 \rangle +2\tilde{\varepsilon}.
\end{equation}
Further, by \eqref{1: hardy assumption} we have
\begin{equation}\label{eq: estimate V phi0 chi2}
\langle V\varphi_1\chi_2,\varphi_1\chi_2 \rangle \leq C\Vert \nabla(\varphi_1 \chi_2) \Vert^2 \leq 2C\Vert (\nabla \varphi_1)\chi_2 \Vert^2+2C\Vert (\nabla \chi_2)\varphi_1 \Vert^2.
\end{equation}
Since $\varphi_1 \in \dot{H}^1(\mathbb{R}^d)$ and $\chi_2$ is bounded and supported in the region $\{x\in \mathbb{R}^d: |x|\geq b\}$, the first term on the r.h.s. of \eqref{eq: estimate V phi0 chi2} is arbitrarily small if $b$ is sufficiently large. Due to \eqref{eq: remark for cmp1} it holds
\begin{equation}
\Vert (\nabla \chi_2)\varphi_1 \Vert^2 \leq \varepsilon \Vert \nabla \varphi_1 \Vert^2=\varepsilon
\end{equation} 
for $\tilde{b}>0$ sufficiently large. This shows that the second term on the r.h.s. of \eqref{eq: estimate V phi0 chi2} can be done arbitrarily small. Hence, we obtain
\begin{equation}\label{eq. estimate for final 2}
\langle V\varphi_1\chi_2,\varphi_1\chi_2 \rangle \leq 2\tilde{\varepsilon}.
\end{equation}
Collecting estimates \eqref{eq: estimate for final 1} and \eqref{eq. estimate for final 2} yields
\begin{equation}\label{eq: estimate <vphi0,phi0>}
\langle V\varphi_1,\varphi_1\rangle \leq \langle V\psi_n\chi_1,\psi_n\chi_1 \rangle + 4\tilde{\varepsilon}
\end{equation}
for $n\in \mathbb{N}$ sufficiently large. 
\\ Let us estimate the r.h.s. of \eqref{eq: estimate <vphi0,phi0>}. Assumption \eqref{1: hardyassumption} implies
\begin{align}
\begin{split}
\langle V\psi_n \chi_1,\psi_n\chi_1 \rangle &= \langle V\psi_n,\psi_n \rangle-\langle V\psi_n \chi_2,\psi_n\chi_2 \rangle\leq \langle V\psi_n,\psi_n \rangle + C\Vert \nabla(\psi_n\chi_2) \Vert^2
\\ &\leq -\left( 1-n^{-1}\right)+C\Vert \nabla(\psi_n\chi_2)\Vert^2.
\end{split}
\end{align}
Due to the remark after Lemma \ref{lemma inside proof 1} we can choose $n\in \mathbb{N}$ and $\tilde{b}>0$, such that $\Vert \nabla(\psi_n\chi_2) \Vert \leq \varepsilon$. Theorefore, we conclude $\langle V\varphi_1 , \varphi_1 \rangle = -1$ and
\begin{equation}\label{qform}
\Vert \nabla \varphi_1 \Vert^2 +\langle V\varphi_1,\varphi_1 \rangle=0.
\end{equation}
Now we prove that $\nabla \left(|x|^{\alpha_0}\varphi_1 \right)\in L^2(\mathbb{R}^d)$ and $(1+|x|)^{\alpha_0-1}\varphi_1 \in L^2(\mathbb{R}^d)$. Let $G_\varepsilon$ be the function defined by \eqref{1: function agmon}. Since $\varphi_1$ is a minimizer of the quadratic form of \eqref{qform} in $\dot{H}^1(\mathbb{R}^d)$, it satisfies the Euler-Lagrange equation in a generalized sense, i.e. it holds
\begin{equation}\label{1: euler lagrange}
\langle \nabla\varphi_1, \nabla \psi \rangle + \langle V\varphi_1,\psi \rangle=0
\end{equation}
for every function $\psi \in \dot{H}^1(\mathbb{R}^d)$. By setting $\psi=G_\varepsilon^2\varphi_1$ we obtain
\begin{equation}\label{111uuu}
\langle \nabla\varphi_1, \nabla \left(G_\varepsilon^2\varphi_1\right) \rangle + \langle V\varphi_1, G_\varepsilon^2\varphi_1 \rangle=0.
\end{equation}
Computations similar to \eqref{eq: agmon Re} together with \eqref{111uuu} yield
\begin{equation}\label{eq: agmon rewritten}
\Vert \nabla(\varphi_1G_\varepsilon) \Vert^2 - \int|\varphi_1|^2 |\nabla G_\varepsilon|^2\, \mathrm{d}x+\int V|\varphi_1G_\varepsilon|^2\, \mathrm{d}x=0.
\end{equation}
By \eqref{1: estimate G} we can rewrite \eqref{eq: agmon rewritten} as
\begin{equation}
\Vert \nabla(\varphi_1G_\varepsilon) \Vert^2 +\langle V \varphi_1 G_\varepsilon, \varphi_1 G_\varepsilon \rangle - \alpha_0^2 \int_{\{|x|\geq 2R\}}\frac{|G_\varepsilon \varphi_1|^2}{|x|^{2}}\, \mathrm{d}x 
\leq \int_{\{R\leq |x|\leq 2R\}} |\varphi_1|^2|\nabla G_\varepsilon|^2\, \mathrm{d}x.
\end{equation}
Since the function $|\nabla G_\varepsilon|$ is uniformly bounded in $\varepsilon$ for $|x|\in[R,2R]$, we have
\begin{align}
\int_{\{R\leq |x|\leq 2R\}} |\varphi_1|^2|\nabla G_\varepsilon|^2\, \mathrm{d}x &\leq C_0 \int_{\{R\leq |x|\leq 2R\}} |\varphi_1 |^2 \, \mathrm{d}x \leq  C_1R^2\int \frac{|\varphi_1|^2}{|x|^2}\, \mathrm{d}x \notag
\\ &\leq C_2 \int |\nabla \varphi_1|^2 \, \mathrm{d}x \leq C_2,
\end{align}
where the constant $C_2>0$ does not depend on $\varepsilon>0$. Similar to the proof of Lemma \ref{1: lem uniformly bound of psi n}, assumption \eqref{1: assumption for agmon} implies
\begin{equation*}
\Vert \nabla (\varphi_1 G_\varepsilon) \Vert \leq C.
\end{equation*}
Taking $\varepsilon \rightarrow 0$ yields $\Vert \nabla(|x|^{\alpha_0}\varphi_1) \Vert<\infty$, which together with Hardy's inequality implies 
\begin{equation}
(1+|x|)^{\alpha_0-1}\varphi_1 \in L^2(\mathbb{R}^d).
\end{equation}
This completes the proof.
\end{proof}
\begin{lem}\label{210}
Assume that
\begin{equation}\label{bbb}
\Vert V \psi \Vert^2 \leq C\left( \Vert \nabla \psi \Vert^2 + \Vert \psi \Vert^2 \right)
\end{equation}
holds for some $C>0$ and every function $\psi\in C_0^\infty(\mathbb{R}^d)$. Then the solution $\varphi_1 \in \dot{H}^1(\mathbb{R}^d)$ in Lemma \ref{1Lem 3} is unique. Moreover, there exists a constant $\delta_1>0$, such that for any function $\psi \in \dot{H}^1(\mathbb{R}^d)$ with $\langle \nabla\psi,\nabla \varphi_1\rangle =0$ it holds 
\begin{equation}\label{1: distance to next virtual level}
\langle h \psi,\psi \rangle \geq \delta_1 \Vert \nabla \psi \Vert^2.
\end{equation}
\end{lem}
\begin{proof}
We will prove the Lemma by contradiction. Assume that there is no such constant $\delta_1>0$, then there exists a sequence of functions $(\psi_n^{(1)})_{n\in\mathbb{N}}$ in $\dot{H}^1(\mathbb{R}^d)$, such that 
\begin{equation}
\Vert \nabla \psi_n^{(1)} \Vert =1, \quad \langle \nabla \psi_n^{(1)},\nabla \varphi_1 \rangle=0,  \quad \left(1-n^{-1}\right)\Vert \nabla \psi_n^{(1)} \Vert^2 + \langle V\psi_n^{(1)},\psi_n^{(1)} \rangle<0.
\end{equation}
Moreover, there exists a subsequence (which by abuse of notation is denoted by $(\psi_n^{(1)})_{n\in\mathbb{N}}$ again) and a function $\tilde{\varphi}_1 \in \dot{H}^1(\mathbb{R}^d)$, such that $\psi_n^{(1)}\rightharpoonup \tilde{\varphi}_1$ in $\dot{H}^1(\mathbb{R}^d)$ and therefore $\psi_n^{(1)} \rightarrow \tilde{\varphi}_1$ in $L_{\mathrm{loc}}^2(\mathbb{R}^d)$. Obviously, $\varphi_1$ and $\tilde{\varphi}_1$ are linearly independent and $\tilde{\varphi}_1$ is a minimizer of the quadratic form of $h$ as well.  Since \eqref{1: euler lagrange} holds for $\psi=\tilde{\varphi}_1$, any linear combination of $\varphi_1$ and $\tilde{\varphi}_1$ is also a minimizer of the quadratic form of $h$. By Hardy's inequality both functions $\varphi_1$ and $\tilde{\varphi}_1$ belong to the weighted $L^2$-space with weight $(1+|\cdot|)^{-2}$. Since the subspace of linear combinations of $\varphi_1$ and $\tilde{\varphi}_1$ is two-dimensional, it contains two functions orthogonal with respect to the weighted scalar product. At least one of these functions, say $f$, has a nontrivial positive part $f_+$ and a nontrivial negative part $f_-$, which are also minimizers of the quadratic form of the operator $h$ and satisfy the corresponding Schr\"odinger equation. Functions $f_+$ and $f_-$ are zero on some open sets. Since $V$ satisfies \eqref{bbb}, the unique continuation Theorem \cite[Theorem 2.1]{schechter} yields $f_+=f_-=0$. This contradiction completes the proof of statement \textbf{(ii)} of Theorem \ref{thm: Theorem 1}.
\end{proof}

The proof of statement \textbf{(iii)} is similar to the proof of Lemma \ref{1: lem uniformly bound of psi n} and Lemma \ref{1Lem 1} with replacing the function $G_\varepsilon$ in \eqref{1: function agmon} by the function
\begin{equation}\label{agmon with exp}
J_{\varepsilon}=\mathrm{exp}\left(\alpha_0\kappa^{-1}\frac{|x|^\kappa}{1+\varepsilon|x|^\kappa}\right)\chi_R(x)
\end{equation}
with $\chi_R(x)$ defined by \eqref{11chi}. This completes the proof of Theorem \ref{thm: Theorem 1}. \qed
\section{Resonances and eigenfunctions on subspaces with fixed symmetries}
Let $h=-\Delta+V$ be invariant under action of a symmetry group $G$ and let $\sigma$ be a type of irreducible representation of $G$. Denote by $P^\sigma$ the projection in $L^2(\mathbb{R}^d)$ onto the subspace of functions transformed according to the representation $\sigma$. In the following we assume that for every function $\psi \in L^2(\mathbb{R}^d)$ and $\chi \in C_0(\mathbb{R}^d)$ with $\chi(x)=\chi(|x|)$ the condition $P^\sigma \psi =\psi$ implies $P^\sigma \chi \psi = \chi \psi$. We denote $h^{\sigma}=P^{\sigma} h$, $h_\varepsilon^{\sigma}=P^{\sigma} h_\varepsilon$, $\mathcal{H}^{\sigma}=P^{\sigma}H^1(\mathbb{R}^d)$ and $\dot{\mathcal{H}}^{\sigma}=P^{\sigma}\dot{H}^1(\mathbb{R}^d)$.
\begin{thm}\label{thm: Theorem 1'}
Suppose that $V$ satisfies \eqref{eq: A1}. Further, assume that
\begin{equation}
h^\sigma\geq 0 \qquad \text{and} \qquad  \inf \mathcal{S}\left( h_\varepsilon^\sigma\right) <0
\end{equation}
holds for any $\varepsilon \in (0,1)$. If there exist constants $\alpha_0>0$, $b>0$ and $\gamma_0\in(0,1)$, such that for any function $\psi\in\mathcal{H}^{\sigma}$ with  $\supp \psi \subset \{x\in \mathbb{R}^d:\ |x|\geq b\}$ we have
\begin{equation}\label{1': assumption for agmon}
\langle h^\sigma\psi,\psi \rangle- \gamma_0\Vert \nabla \psi \Vert^2 - \langle \alpha_0^2|x|^{-2}\psi,\psi \rangle \geq 0,
\end{equation}
then the following assertions hold:
\begin{enumerate}
\item[\textbf{(i)}] If $\alpha_0>1$, then zero is an eigenvalue of $h^\sigma$ with finite degeneracy. Denote by $\mathcal{W}_0$ the corresponding eigenspace. Then for any $\varphi_0 \in \mathcal{W}_0$ we have 
\begin{equation}
\nabla \left(|x|^{\alpha_0} \varphi_0 \right) \in L^2(\mathbb{R}^d) \qquad \text{and} \qquad (1+|x|)^{\alpha_0-1}\varphi_0 \in L^2(\mathbb{R}^d).
\end{equation}
Moreover, there exists a constant $\delta_0>0$, such that for any function $\psi \in \mathcal{H}^{\sigma}$ with $\langle \nabla\psi,\nabla \varphi_0\rangle =0$ for all $\varphi_0 \in \mathcal{W}_0$ it holds
\begin{equation}
\langle h^\sigma \psi,\psi \rangle \geq \delta_0 \Vert \nabla \psi \Vert^2.
\end{equation}
\item[\textbf{(ii)}] If $\alpha_0 \in (0,1)$ and in addition
\begin{equation}
\langle |V|\psi,\psi \rangle \leq C\Vert \nabla \psi\Vert^2
\end{equation}
holds for any function $\psi \in \dot{\mathcal{H}}^{\sigma}$ and some constant $C>0$, then there exists a finite-dimensional subspace $\mathcal{W}_1\subset \dot{\mathcal{H}}^{\sigma}$, such that for any function $\varphi_1 \in \mathcal{W}_1$ it holds
\begin{equation}
\Vert \nabla \varphi_1 \Vert^2+ \langle V\varphi_1,\varphi_1 \rangle = 0.
\end{equation}
Moreover, it holds 
\begin{equation}
\nabla \left(|x|^{\alpha_0} \varphi_1 \right) \in L^2(\mathbb{R}^d) \qquad \text{and} \qquad (1+|x|)^{\alpha_0-1}\varphi_1 \in L^2(\mathbb{R}^d)
\end{equation}
and there exists a constant $\delta_1>0$, such that for any function $\psi \in \dot{\mathcal{H}}^{\sigma}$ satisfying the condition $\langle \nabla\psi,\nabla \varphi_1\rangle =0$ for all $\varphi_1 \in \mathcal{W}_1$ we have
\begin{equation}
\langle h^\sigma \psi,\psi \rangle \geq \delta_1 \Vert \nabla \psi \Vert^2.
\end{equation}
\item[\textbf{(iii)}] If instead of \eqref{1': assumption for agmon} a stronger inequality
\begin{equation}
\langle h^\sigma\psi,\psi\rangle - \gamma_0 \Vert \nabla \psi \Vert^2 - \langle \alpha_0^2|x|^{-\beta} \psi,\psi \rangle \geq 0
\end{equation}
holds for some constant $\alpha_0>0$ and $\beta\in(0,2)$, then each function $\varphi_0 \in \mathcal{W}_0$ in part \textbf{(i)} of the theorem satisfies
\begin{equation}
\exp\left( \alpha_0\kappa^{-1}|x|^{\kappa} \right)\varphi_0 \in L^2(\mathbb{R}^d), \qquad \text{where}\quad \kappa=1-\frac{\beta}{2}.
\end{equation}
\end{enumerate}
\end{thm}
\begin{proof}
The proof of Theorem \ref{thm: Theorem 1'} is a straightforward generalization of the proof of Theorem \ref{thm: Theorem 1}. The main difference between these two theorems is that in Theorem \ref{thm: Theorem 1} we have non-degenerate minimizers $\varphi_0$ or $\varphi_1$ of the quadratic form of the operator $h$ in the spaces $H^1(\mathbb{R}^d)$ and $\dot{H}^1(\mathbb{R}^d)$, respectively. In Theorem \ref{thm: Theorem 1'} the corresponding subspaces $\mathcal{W}_0$ and $\mathcal{W}_1$ are not necessarily one-dimensional. However, due to Lemma \ref{A3} (see in Appendix) they are always finite-dimensional.
\end{proof}
\begin{rem}
Theorem \ref{thm: Theorem 1} and Theorem \ref{thm: Theorem 1'} require $d\geq 3$. We used this condition twice. At first, we used Hardy's inequality to compensate the localization error $\varepsilon |x|^{-2}$ with a part of the kinetic energy in Lemma \ref{eq: cmp1}. Secondly, we used the Rellich–Kondrachov theorem in the proof of Theorem \ref{thm: Theorem 1} to obtain convergence of the constructed subsequence in $L_{\mathrm{loc}}^2(\mathbb{R}^d)$. If the dimension is one or two, but the operator $h$ is considered on a subspace with a fixed symmetry $\sigma$, such that Hardy's inequality
\begin{equation}
\Vert \nabla \psi \Vert^2 \geq C\Vert \psi |x|^{-1} \Vert^2
\end{equation}
holds for some $C>0$, the statement of Theorem \ref{thm: Theorem 1'} remains true.
\end{rem}
\section{Applications}
\subsection{Virtual levels of one-body Schr\"odinger operators}
The main goal of our paper is to study decay properties of virtual levels of multi-particle Schr\"odinger operators. However, in order to show how effective Theorem \ref{thm: Theorem 1} is we start with the easiest case of one-particle Schr\"odinger operators. Some of the results below are already known. 
\\ 
For $\varepsilon\in (0,1)$ we consider
\begin{equation}
h=-\Delta+V \qquad \qquad \text{and} \qquad \qquad  h_\varepsilon=h+\varepsilon\Delta
\end{equation}
in $L^2(\mathbb{R}^d),$ where $d\geq 3$ and $V$ satisfies \eqref{eq: A1}.
\begin{thm}\label{aaaa}(Short-range potentials)
Let $d\geq 3, h\geq 0$ and assume that for any sufficiently small $\varepsilon>0$ we have
\begin{equation}
\inf \mathcal{S}_{\mathrm{ess}}(h_\varepsilon)=0 \qquad \text{and} \qquad \inf \mathcal{S}(h_\varepsilon)<0.
\end{equation}
Further, assume that one of the following conditions is fulfilled:
\begin{enumerate}
\item[\textbf{(i)}] $d=3$ and $V\in L^{\frac{3}{2}}(\mathbb{R}^3)$,
\item[\textbf{(ii)}] $d=4$ and $V\in L^2(\mathbb{R}^4)\cap L^{2+\mu}(\mathbb{R}^4)$ for some $\mu>0$,
\item[\textbf{(iii)}] $d\geq 5$ and $V\in L^{\frac{d}{2}}(\mathbb{R}^d)$.
\end{enumerate}
Then there exists a solution $\varphi_0 \in \dot{H}^1(\mathbb{R}^d)$ of the equation
\begin{equation}
\Vert \nabla \varphi_0 \Vert^2+\langle V\varphi_0,\varphi_0\rangle =0.
\end{equation}
For any $0\leq \alpha_0 < \frac{d-2}{2}$ function $\varphi_0$ satisfies
\begin{equation}\label{ppp}
\nabla \left( |x|^{\alpha_0}\varphi_0 \right)\in L^2(\mathbb{R}^d) \qquad \text{and} \qquad (1+|x|)^{\alpha_0-1} \varphi_0 \in L^2(\mathbb{R}^d).
\end{equation}
\end{thm}
\begin{rem}
Theorem \ref{aaaa} implies in particular that for $d\geq 5$ virtual levels of $h$ are eigenvalues.
\end{rem}
\begin{proof}
According to \cite{reed}, p.170-171 and Sobolev's inequality the potential $V$ satisfies \eqref{eq: A1} and \eqref{1: hardyassumption}. Moreover, conditions \textbf{(i)}-\textbf{(iii)} and Sobolev's inequality imply that for any $\varepsilon>0$ and sufficiently large $b>0$ 
\begin{equation}
\varepsilon \Vert \nabla \varphi \Vert^2 + \langle V\varphi,\varphi\rangle \geq 0
\end{equation}
holds for any $\varphi \in \dot{H}^1(\mathbb{R}^d)$ with $\supp \varphi \subset \{x \in \mathbb{R}^d : \ |x|\geq b\}$. Applying Hardy's inequality we see that condition \eqref{1: assumption for agmon} of Theorem \ref{thm: Theorem 1} is fulfilled. This yields the result.
\end{proof}

\begin{thm}\label{a}(Long-range potentials positive at infinity)
Let $d\geq3, h\geq 0$ and assume that for any sufficiently small $\varepsilon>0$ we have
\begin{equation}
\inf \mathcal{S}_{\mathrm{ess}}(h_\varepsilon)=0 \qquad \text{and} \qquad \inf \mathcal{S}(h_\varepsilon)<0.
\end{equation}
Further, assume that
\begin{enumerate}
\item[\textbf{(i)}] $V\in L_{\mathrm{loc}}^{\frac{d}{2}}(\mathbb{R}^d)$ for $d\not = 4$ and $V\in L_{\mathrm{loc}}^{2+\mu}(\mathbb{R}^d)$ for some $\mu>0$ if $d = 4$.
\item[\textbf{(ii)}] There exist constants $A_1,A_2\geq 0,\beta_1> 0$ and $\beta_2\in(0,2]$ with
\begin{equation}
\beta_1|x|^{-\beta_2} \leq V(x)\leq A_1 \qquad \text{for} \qquad |x|\geq A_2.
\end{equation}
\item[\textbf{(iii)}] $V(x)\rightarrow 0$ as $|x|\rightarrow \infty$.
\end{enumerate}
Then there exists a solution $\varphi_0 \in \dot{H}^1(\mathbb{R}^d)$ of the equation
\begin{equation}
\Vert \nabla \varphi_0 \Vert^2+\langle V\varphi_0,\varphi_0\rangle =0.
\end{equation}
If $\beta_2=2$, then for any $0\leq \alpha_0< \sqrt{\beta_1+4^{-1}(d-2)^2}$ the function $\varphi_0$ satisfies
\begin{equation}\label{ppp}
\nabla \left( |x|^{\alpha_0}\varphi_0 \right)\in L^2(\mathbb{R}^d) \qquad \text{and} \qquad (1+|x|)^{\alpha_0-1} \varphi_0 \in L^2(\mathbb{R}^d).
\end{equation}
If $\beta_2<2$, then $\varphi_0$ satisfies
\begin{equation}
\exp\left( \beta_1|x|^\kappa\right)\varphi_0\in L^2(\mathbb{R}^d), \qquad \text{where} \quad \kappa=1-\frac{\beta_2}{2}.
\end{equation}
\end{thm}
\begin{rem}
Theorem \ref{a} implies in particular that for $d=3$ zero is an eigenvalue of $h$ for $\beta_1>\frac{3}{4}$ and in case $d=4$ zero is an eigenvalue of $h$ for any $\beta_1>0$.
\end{rem}

\subsection{Virtual levels of $N$-body Schr\"odinger operators}
Now we consider a system of $N\geq 3$ quantum particles in dimension $n\geq 3$ with masses $m_i>0, \ i =1,\dots,N$ and position vectors $x_i \in \R^n, \ i=1,\dots, N$. Such a system is described by the Hamiltonian
\begin{equation}\label{Definition Hamiltonian}
H_N = -\sum\limits_{i=1}^N \frac{1}{m_i} \Delta_{x_i} +  \sum\limits_{1\leq i<j \leq N} V_{ij}(x_{ij}), \quad x_{ij}=x_i-x_j
\end{equation}
acting on $L^2\left(\R^{nN}\right)$. The potentials $V_{ij}$ describe the particle pair interaction and in the following we assume that $V_{ij}=V_{ij}^{(1)}+V_{ij}^{(2)}$, such that for some constants $A,C,\nu>0$ we have
\begin{equation}\label{2: assumptions on potential} 
|V_{ij}^{(1)}(x_{ij})| \leq C|x_{ij}|^{-2-\nu}, \ \text{if}\ |x_{ij}|\geq A \qquad \text{and} \qquad V_{ij}^{(1)} \in L_{\mathrm{loc}}^{p}(\mathbb{R}^n),
\end{equation}
with $p>2$ for $n=4$ and $p=\frac{n}{2}$ for $n\neq 4$. Furthermore, we assume that
\begin{equation}\label{2: assumptions on potential part 2} 
V_{ij}^{(2)}\geq 0\quad \text{is bounded and}\quad V_{ij}^{(2)}(x_{ij})\rightarrow 0 \quad \text{as}\quad|x_{ij}|\rightarrow \infty.
\end{equation}
Under these conditions $V_{ij}$ is relatively form-bounded with relative bound zero, i.e. it satisfies \eqref{eq: A1}, see \cite{reed}, p. 170-171.
\subsubsection*{Separation of the center of mass of the system}
We will consider the operator $H_N$ in the center-of-mass frame following \cite{Sigalov}. We introduce the scalar product $\langle \cdot, \cdot\rangle_m$ on $\mathbb{R}^{nN}$ by
\begin{equation}\label{scalar 1}
    \langle x,y\rangle_m = \sum\limits_{i=1}^N m_i \langle x_i ,y_i \rangle, \qquad \vert x\vert^2_m = \langle x,x\rangle_m, \qquad x, y \in \R^{nN}.
\end{equation}
Here we denote by $\langle \cdot, \cdot \rangle $ the standard scalar product on $\R^{n}$. Let $X$ be the space $\mathbb{R}^{nN}$ equipped with the scalar product $\langle \cdot,\cdot \rangle_m$ and let
\begin{equation}\label{R_0[C]}
    X_0 = \left\{x=(x_1,\ldots,x_N) \in X \; : \; \sum\limits_{i=1}^N m_i x_i =0 \right\} 
\end{equation}
be the space of relative positions of the particles and $X_c =X \ominus X_0$ be the space of the center of mass position of the system. Denote by $P_0$ and $P_c$ the corresponding projections from $X$ on $X_0$ and $X_c$, respectively.
\\Furthermore, we introduce $-\Delta, -\Delta_0$ and $-\Delta_c$ as the Laplace-Beltrami operators with respect to \eqref{scalar 1} on $X, X_0$ and $X_c$, respectively. Then, corresponding to $L^2(X)= L^2(X_0)\otimes L^2(X_c)$ we have
\begin{equation}
-\Delta=-\Delta_0 \otimes I+I\otimes (-\Delta_c).
\end{equation}
Since for every $x\in X$ we have
\begin{equation}
(P_0 x)_i-(P_0 x)_j =x_i-x_j,
\end{equation}
it follows that the potential $V(x)=\sum_{1\leq i<j\leq N} V_{ij}(x_{ij})$ satisfies
\begin{equation}
V(x)=V(P_0 x).
\end{equation}
Hence, $H_N$ is unitarily equivalent to the operator
\begin{equation}\label{center of mass}
H\otimes I + I\otimes\left(-\Delta_c\right),
\end{equation}
where 
\begin{equation}\label{H_0[C]}
H=-\Delta_0+V.
\end{equation}
In view of \eqref{center of mass} the center of mass of the system moves like a free particle and the Hamiltonian $H$ corresponds to the relative motion of the system. This procedure is known as the reduction of the center of mass of the system. In the following we only consider the operator $H$.
\subsubsection*{Clusters and cluster Hamiltonians} 
We call an arbitrary non-empty subset $C\subseteq \{1,\ldots,N\}$ a cluster of the system and denote by $|C|$ the number of its particles. Let
\begin{equation}\label{R0 of C}
X_0[C]=\left\{x \in X_0 \; : \; \sum\limits_{i\in C} m_i x_i =0, \ x_j =0, \; j \notin C \right\}
\end{equation}
be the corresponding subspace of the relative positions of the particles within the cluster $C$. Let $-\Delta_0[C]$ be the Laplace-Beltrami operator on $X_0[C]$. We denote the potential of interactions of the particles in $C$ by
\begin{equation}
V[C]=\sum_{i,j\in C, i < j}V_{ij}.
\end{equation}
Then for $1<|C|<N$ the corresponding cluster Hamiltonian with its center of mass removed is given by
\begin{equation}\label{HC without}
H[C]=-\Delta_0[C]+V[C].
\end{equation}
The operator $H[C]$ acts on $L^2(X_0[C])$ and it describes the relative motion of the particles within the cluster $C$ ignoring all the other particles of the system. Note that for $C=\{1,\ldots,N\}$ we have $X_0[C]=X_0$, so we set $H[C]=H$. For $|C|=1$ we have $X_0[C]=\{0\}$, so in this case we set $H[C]=0$.
\subsubsection*{Partitions of the system} 
We say that $Z=(C_1,\ldots,C_p)$ is a partition or a cluster decomposition of the system of order $|Z|=p$, if and only if for all $i,j=1,\ldots,p$ with $i\not=j$ we have
\begin{equation}
C_i\cap C_j= \emptyset \qquad \text{and} \qquad \bigcup\limits_{j=1}^p C_j = \{1,\ldots,N\}.
\end{equation}
We refer to $C\subset Z$ as a cluster of the partition $Z=(C_1,\ldots,C_p)$, if $C=C_i$ for some $i=1,\ldots,p$. Let
\begin{equation}\label{R0 of Z}
   X_0(Z)=\bigoplus \limits_{C_k\subset Z} X_0[C_k] \qquad \text{and} \qquad X_c(Z)= X_0\ominus  X_0(Z).
\end{equation}
This gives rise to the decomposition
\begin{equation}
L^2(X_0(Z))=L^2(X_0[C_1])\otimes \dots \otimes L^2(X_0[C_p]).
\end{equation}
By abuse of notation we denote the operators
\begin{equation}\label{H[C] on X(Z)}
I\otimes \dots \otimes I \otimes (-\Delta_0[C_i]) \otimes I \otimes \dots \otimes I\qquad \text{and} \qquad I\otimes \dots \otimes I \otimes H[C_i] \otimes I \otimes \dots \otimes I
\end{equation}
acting on $L^2(X_0(Z))$ by $-\Delta_0[C_i]$ and $H[C_i]$, respectively. Then the cluster decomposition Hamiltonian acting on $L^2(X_0(Z))$ is defined by
\begin{equation}
    H(Z)= \sum\limits_{C_k\subset Z} H[C_k].
\end{equation}
The operator $H(Z)$ describes the joint internal dynamics of the non-interacting clusters. Let $-\Delta_0(Z)$ be the Laplace-Beltrami operator on $X_0(Z)$. Then
\begin{equation}
-\Delta_0(Z)=-\sum_{C_k\subset Z}\Delta_0[C_k].
\end{equation} 
\\Corresponding to the decomposition $L^2(X_0)=L^2(X_0(Z))\otimes L^2(X_c(Z))$ we will sometimes use the same symbols $H[C_i]$ and $H(Z)$ for the operators acting on $L^2(X_0)$ as
\begin{equation}\label{H(Z) on X}
H[C_i]\otimes I \qquad \text{and} \qquad  H(Z)\otimes I,
\end{equation}
respectively. We denote the intercluster interaction by
\begin{equation}
I(Z)=V-\sum_{C_k \subset Z} V[C_k].
\end{equation}
Then the Hamiltonian $H$ can be written as
\begin{equation}
H=H(Z)\otimes I + I\otimes (-\Delta_c(Z))+I(Z),
\end{equation}
where $-\Delta_c(Z)$ is the Laplace-Beltrami operator on $X_c(Z)$. Denote by $P_0(Z)$ and $P_c(Z)$ the corresponding projections from $X_0$ on $X_0(Z)$ and $X_c(Z)$, respectively. For $x\in X_0$ we set
\begin{equation}
q(Z)=P_0(Z)x\qquad \text{and} \qquad \xi(Z)=P_c(Z)x.
\end{equation}
To emphasize the dependence of $q(Z)$ and $\xi(Z)$ we will write
\begin{equation}
-\Delta_{q(Z)} = - \Delta_0(Z) \qquad \text{and} \qquad -\Delta_{\xi(Z)}=-\Delta_c(Z)
\end{equation}
and
\begin{equation}\label{H decomposed in partitions}
H=-\Delta_{q(Z)}-\Delta_{\xi(Z)}+V \qquad \text{or} \qquad H=H(Z)-\Delta_{\xi(Z)}+I(Z).
\end{equation}
Note that the $i$-th coordinates of $q(Z)$ and $\xi(Z)$ are vectors $q_i$ and $\xi_i$ given by
\begin{equation}
q_i=x_i-x_{C_l}\qquad \text{and} \qquad \xi_i=x_{C_l},
\end{equation}
where $C_l$ is the cluster of the partition $Z$ with $i\in C_l$ and $x_{C_l}$ is the center of mass of the cluster $C_l$ given by
\begin{equation}
x_{C_l} = \frac{1}{\sum_{j\in C_l} m_j} \sum_{k\in C_l} m_kx_k.
\end{equation}
\\With regard to $q(Z)$ and $\xi(Z)$ we introduce the following regions, which we will refer to as cones in the following. For $\kappa>0$ and partitions $Z$ with $1<|Z|<N$ let
\begin{align}\label{eq: cones}
    \begin{split}
        K(Z,\kappa) &=\left\{x\in X_0 \ :\ \vert q\left({Z}\right)\vert_m \le \kappa \vert \xi\left(Z\right)\vert_m\right\} .       \\
    \end{split}
\end{align}
For the entire system $Z_1=(\{1,\ldots,N\})$ we set
	\begin{equation}\label{eq: cone for system}
	K(Z_1,\kappa)=\{x\in X_0: |x|_m \leq \kappa\}.
	\end{equation}
\begin{df}\label{2345}
For an arbitrary cluster $C\subseteq \{1,\ldots,N\}$ we say that the operator $H[C]=-\Delta_0[C]+V[C]$ has a virtual level at zero, if $H[C] \ge 0$ and for all sufficiently small $\varepsilon>0$ we have
\begin{equation}\label{ess}
\mathcal{S}_{\mathrm{ess}}\left(-(1-\varepsilon)\Delta_0[C] + V[C]\right)=[0,\infty)
\end{equation}
and
\begin{equation}\label{disc}
\mathcal{S}_{\mathrm{disc}}\left(-(1-\varepsilon)\Delta_0[C] + V[C]\right)\neq \emptyset.
\end{equation}
\end{df}
\begin{rem}
Assume that $H$ has a virtual level. Then condition \eqref{ess} together with the HVZ-theorem implies that there exists $\varepsilon>0$, such that for any non-trivial cluster $C$ with $1<|C|<N$ we have
\begin{equation}\label{ineq for ess}
\mathcal{S}\left( -(1-\varepsilon)\Delta_0[C]+V[C] \right) = [0,\infty).
\end{equation}
In particular \eqref{ineq for ess} yields that if $H$ has a virtual level, then the Hamiltonians corresponding to the non-trivial clusters of the system do not have resonances or eigenvalues at zero.
\end{rem}
The main result of this section is the following

\begin{thm}\label{Main theorem}
Consider a system of $N\ge 3$ particles in dimension $n\ge 3$. Suppose that the potentials $V_{ij}$ satisfy \eqref{2: assumptions on potential} and \eqref{2: assumptions on potential part 2}. Assume that $H$ has a virtual level at zero. Then
\begin{enumerate}
\item[\textbf{(i)}] zero is an eigenvalue of $H$ and the corresponding eigenfunction $\varphi_0$ satisfies
\begin{equation}\label{alpha 0 in this case}
\nabla_0\left(|x|_m^{\alpha_0}\varphi_0 \right)\in L^2(X_0) \qquad \text{and} \qquad (1+|x|_m)^{\alpha_0-1}\varphi_0\in L^2(X_0)
\end{equation}
for any $0\leq \alpha_0 < \frac{n(N-1)-2}{2}$.
\item[\textbf{(ii)}] There exists a constant $\delta_0>0$, such that for every function $\psi\in H^1(X_0)$ satisfying $\langle \nabla_0 \psi, \nabla_0 \varphi_0 \rangle=0$ we have
\begin{equation}
(1-\delta_0)\Vert \nabla_0 \psi \Vert^2 + \langle V\psi,\psi \rangle \geq 0.
\end{equation} 
\item[\textbf{(iii)}] If $V_{ij}^{(2)}$ satisfies $V_{ij}^{(2)}(x)\geq \alpha_{ij}|x|^{-\beta}$ for some constants $\alpha_{ij}>0$ and $\beta\in(0,2)$, then zero is an eigenvalue of $H$ and the corresponding eigenfunction $\varphi_0$ satisfies
\begin{equation}
\mathrm{exp}\left(\mu|x|_m^\kappa\right)\varphi_0 \in L^2(X_0),
\end{equation}
where $\kappa=1-\frac{\beta}{2}$ and $\mu>0$ depends on the coefficients $\alpha_{ij}$ and on the masses of the particles only.
\end{enumerate}
\end{thm}
\begin{rem}
\begin{enumerate}
\item[$\textbf{(i)}$] Theorem \ref{Main theorem} tells us that for $n$-dimensional particles with $n\geq 3$ virtual levels may be resonances for two-body Hamiltonians only. This is the reason why the Efimov effect does not occur for $n\geq3$ and $N\geq 4$. The proof of the absence of the Efimov effect will be given in Theorem \ref{No Efimov} in the next section.
\item[$\textbf{(ii)}$] Note that, as it is usual for variational methods for multi-particle Schr\"odinger operators, Theorem \ref{Main theorem} has very weak restrictions on the decay of the positive part of the potentials. Part \textbf{(iii)} of the theorem shows that if the interactions of particles for large distances are long-range and positive, an Agmon-type method can be used to prove the sub-exponential decay of eigenfunctions at the bottom of the essential spectrum. This idea is due to D. Hundertmark, M. Jex and M. Lange, see \cite{hundert1,hundert2}.
\end{enumerate}
\end{rem}
Before proving the theorem, we will generalize it in two directions: We will give an analogue of this theorem for systems including a particle of infinite mass and we will consider systems with symmetry restrictions.

\subsubsection{Remark on $N$-particle systems in an external potential field}
Now we consider the operator $\tilde{H}_N$ of an $N$-particle system $\{1,\ldots,N\}$ with an external electric field
\begin{equation}
\tilde{H}_N=-\sum_{i=1}^N \frac{1}{m_i}\Delta_{x_i} + \sum_{1\leq i<j\leq N} V_{ij}(x_{ij})+\sum_{j=1}^N U_j(x_j)
\end{equation}
acting on $L^2(\mathbb{R}^{nN})$. 
\\Formally, this system can be considered as an $(N+1)$-particle system with one particle having infinite mass and positioned at the origin. This particle will have number zero and we assume that the external potentials $U_j$ satisfy the same conditions as $V_{ij}$. Similar to \eqref{scalar 1} we define the scalar product $\langle \cdot,\cdot\rangle_m$ in $\mathbb{R}^{nN}$ and write the operator $\tilde{H}_N$ as
\begin{equation}
H_\infty =-\Delta_0+\sum_{1\leq i<j\leq N}V_{ij} +\sum_{j=1}^N U_j,
\end{equation}
where $-\Delta_0$ is the Laplace-Beltrami operator on $X$. Note that for the operator $H_\infty$ we do not need to separate the center of mass motion.
\\Let $H_\infty \geq 0$. We say that $H_\infty$ has a virtual level if for all sufficiently small $\varepsilon>0$ we have
\begin{equation}
\mathcal{S}_{\mathrm{ess}}(H_\infty +\varepsilon\Delta_0)=[0,\infty)
\end{equation}
and
\begin{equation}
\mathcal{S}_{\mathrm{disc}}(H_\infty+\varepsilon\Delta_0)\not = \emptyset.
\end{equation}

With the definitions given above assertions \textbf{(i)}-\textbf{(iii)} of Theorem \ref{Main theorem} hold for $N+1$ with $H$ replaced by $H_\infty$ and $X_0$ replaced by $X$.

\subsubsection{Systems with permutational symmetry}\label{sec}
Assume now that a system of $N$ particles contains several identical particles of finite mass. Let $S$ be the group of permutations of identical particles and $\sigma$ be a type of irreducible representation of this group. Let $P^\sigma$ be the corresponding projection on the subspace of functions transformed according to the representation $\sigma$. For any fixed partition $Z=(C_1,\ldots,C_p), \ 2\leq p \leq N-1,$ we define $S(Z)$ as a group, which permutes identical particles within the clusters $C_k \subset Z, \ k=1,2,\ldots,p$ and permutes identical clusters if such clusters exist in $Z$. Obviously $S(Z)$ is a subgroup of $S$. Denote by $\sigma'(Z)$ types of irreducible representations of $S(Z)$. We say that the representation $\sigma'(Z)$ of the group $S(Z)$ is induced by the representation $\sigma$ of the group $S$ and write $\sigma'(Z) \prec \sigma$ if $\sigma'(Z)$ is contained in $\sigma$ restricted to $S(Z)$.
\begin{df}\label{234}
We say that $H^\sigma:=P^\sigma H$ has a virtual level of symmetry $\sigma$, if $H^\sigma \geq 0$ and for all sufficiently small $\varepsilon>0$ it holds
\begin{equation}\label{989}
\mathcal{S}_{\mathrm{ess}}\left( P^{\sigma}(H+\varepsilon\Delta_0)\right)=[0,\infty)
\end{equation}
and
\begin{equation}
\mathcal{S}_{\mathrm{disc}}\left( P^\sigma(H+\varepsilon\Delta_0) \right) \not = \emptyset.
\end{equation}
\end{df}

\begin{rem}
Analogously to the remark after Definition \ref{2345}, condition \eqref{989} together with the HVZ-Theorem \cite{hhvz} imply that for any partition $Z=(C_1,\ldots,C_p)$ with $1<p<N$ and any type of irreducible representation $\sigma'(Z)\prec \sigma$ it holds
\begin{equation}
P^{\sigma'(Z)} \left( H(Z)+\varepsilon\Delta_{q(Z)} \right) \geq 0
\end{equation}
for sufficiently small $\varepsilon>0$.
\end{rem}

\begin{thm}\label{1: thm with symmetry}
Suppose that $N\geq 3$ and consider the operator $H^\sigma$, where the potentials $V_{ij}$ satisfy \eqref{2: assumptions on potential} and \eqref{2: assumptions on potential part 2}. Assume that $H^\sigma$ has a virtual level of symmetry $\sigma$. Then
\begin{enumerate}
\item[\textbf{(i)}] zero is an eigenvalue of $H^\sigma$ with finite degeneracy. Let $\mathcal{W}_0$ be the corresponding eigenspace, then for any $\varphi_0 \in \mathcal{W}_0$ we have 
\begin{equation}
\nabla_0 \left( |x|_m^{\alpha_0} \varphi_0 \right) \in L^2(X_0) \qquad \text{and} \qquad (1+|x|_m)^{\alpha_0-1}\varphi_0 \in L^2(X_0)
\end{equation}
for any $0\leq \alpha_0 < \frac{n(N-1)-2}{2}.$
\item[\textbf{(ii)}] There exists a constant $\delta_0>0$, such that for any function $\psi \in P^\sigma H^1(X_0)$ satisfying $\langle \nabla_0\psi,\nabla_0 \varphi_0\rangle =0$ for all $\varphi_0 \in \mathcal{W}_0$, it holds
\begin{equation}
(1-\delta_0)\Vert \nabla_0 \psi \Vert^2 + \langle V\psi,\psi \rangle\geq 0.
\end{equation}
\item[\textbf{(iii)}] If $V_{ij}^{(2)}$ satisfies $V_{ij}^{(2)}(x)\geq \alpha_{ij}|x|^{-\beta}$ for some constants $\alpha_{ij}>0$ and $\beta\in(0,2)$, then for every function $\varphi_0 \in \mathcal{W}_0$ we have
\begin{equation}\label{statement three}
\mathrm{exp}\left(\mu|x|_m^\kappa\right) \varphi_0 \in L^2(X_0),
\end{equation}
where $\kappa=1-\frac{\beta}{2}$ and $\mu>0$ depends on the coefficients $\alpha_{ij}$ and on the masses of the particles only.
\end{enumerate}
\end{thm}
\begin{rem}
The decay rate of the eigenfunctions $\varphi_0 \in \mathcal{W}_0$ depends on the corresponding Hardy constant $c_H$, which on the whole space $L^2(X_0)$ is given by $c_H=\frac{(\dim X_0-2)^2}{4}$. However, if $\sigma$ is a representation different from the symmetric representation, the Hardy constant can become larger. This can result in a stronger decay rate of the eigenfunctions.
\end{rem}
\section*{Proof of Theorem \ref{Main theorem}}
The following estimate for the localization error, originally proved in \cite{Semjon2}, plays a crucial role in the proof of Theorems \ref{Main theorem} and \ref{1: thm with symmetry}. For the convenience of the reader a complete proof of this estimate is given in the Appendix.
\begin{lem}\cite[Lemma 5.1]{Semjon2}\label{Lemma localization error} Given $\varepsilon>0$ and $\kappa>0$, for each partition $Z$ with $1<|Z|<N$ one can find $0 < \kappa'<\kappa$ and functions $u_{Z},\ v_{Z}\ : \ X_0\rightarrow\R$, such that
\begin{equation}\label{Localization functions}
u_{Z}^2+v_{Z}^2 =1, \qquad u_{Z}(x)=\begin{cases}
	1, & x\in K\left(Z,\kappa'\right)\\
	0, & x\notin K\left(Z,\kappa\right)
\end{cases}
\end{equation}
and
\begin{equation}\label{Localization error I}
	\vert \nabla_0 u_{Z} \vert^2 + \vert \nabla_0 v_{Z} \vert^2 < \varepsilon \left[\vert v_{Z} \vert^2\vert x\vert_m^{-2}+\vert u_{Z} \vert^2\vert q\left(Z\right)\vert_m^{-2}\right]
\end{equation}
for $x\in  K\left(Z,\kappa\right)\setminus K\left(Z,\kappa'\right)$.
\end{lem}
To explain the main ideas of the proof of Theorem \ref{Main theorem} we start with $N=3$ and extend the strategy to the case $N\geq 4$ afterwards.
\begin{proof}[Proof of Theorem \ref{Main theorem} for $N=3$ particles and $n=3$] 
Note that in this case we have to prove \eqref{alpha 0 in this case} with $\alpha_0\in(0,2)$. We will prove that all conditions of statement \textbf{(i)} of Theorem \ref{thm: Theorem 1} are fulfilled. We will also show that if in addition $V_{ij}^{(2)}(x) \geq \alpha_{ij}|x|^{-\beta}$ holds for some constants $\alpha_{ij}> 0$ and $\beta \in(0,2)$, then \eqref{statement three} follows from statement \textbf{(iii)} of Theorem \ref{thm: Theorem 1}. 
\\Since $V_{ij} \in L_{\mathrm{loc}}^{\frac{3}{2}}(\mathbb{R}^3)$ and it decays at infinity, for any $\varepsilon>0$ there exists a constant $C(\varepsilon)>0$, such that
\begin{equation}
\langle |V_{ij}|\varphi,\varphi \rangle \leq \varepsilon \Vert \nabla_{x_{ij}} \varphi \Vert^2 + C(\varepsilon) \Vert \varphi \Vert^2
\end{equation}
holds for any function $\varphi \in C_0^\infty (X_0)$. This obviously implies \eqref{eq: A1} for $V=\sum_{1\leq i < j \leq N}V_{ij}$, see \cite{reed}, p.170. 
\\To prove statements \textbf{(i)} and \textbf{(ii)} of the theorem it is sufficient to prove that
\begin{equation}\label{Quadratic form to be estimated}
L[\varphi]:=(1-\gamma_0)\left\Vert\nabla_0 \varphi\right\Vert^2+ \left\langle V\varphi,\varphi \right\rangle - \Vert \alpha_0 \vert x\vert_m^{-1}\varphi\Vert^2\ge 0
\end{equation}
holds for some constant $\gamma_0>0$, any $\alpha_0 \in (1,2)$ and any function $\varphi\in H^1(X_0)$, which satisfies $\supp (\varphi) \subset \{x\in X_0 : |x|_m \geq R\}$ for some sufficiently large $R>0$.
\\The proof of \eqref{Quadratic form to be estimated} follows the ideas of the estimate from below of the quadratic form of a multi-particle Schr\"odinger operator in \cite{Semjon1} in the easiest case when the corresponding cluster Hamiltonians do not have bound states or virtual levels. The difference between \eqref{Quadratic form to be estimated} and a similar inequality proved in \cite{Semjon1} is that for the purposes of \cite{Semjon1} it was sufficient to prove this inequality with an arbitrary small $\alpha_0>0$. Now we need to prove \eqref{Quadratic form to be estimated} with $\alpha_0 \in(1,2)$. Following \cite{Sem4} we will make a partition of unity of the support of the function $\varphi$, separating regions corresponding to different partitions of the system into clusters.
\\Let $u_{Z}$ be the localization functions defined by \eqref{Localization functions}. Recall that for $|Z|=2$ the function $u_{Z}$ is supported in the cone in the configuration space, where two particles belonging to the same cluster in $Z$ are close one to another and the third particle is very far away from this cluster. Due to Theorem \ref{Lemma Intersection cones} in the in the Appendix we can choose $\kappa>0$ sufficiently small, such that the cones $K(Z,\kappa)$ for different $Z$ with $|Z|=2$ do not overlap on the support of $\varphi$. Then Lemma \ref{Lemma localization error} yields
 \begin{equation}
 	L[\varphi] \ge \sum\limits_{Z,|Z|=2} L_1\left[\varphi u_{Z}\right]+L_2\left[\varphi \mathcal{V}\right],
 \end{equation}
where $\mathcal{V}= \sqrt{1-\sum_{Z,|Z|=2} u^2_{Z}}$ and the functionals $L_1,L_2: H^1(X_0)\rightarrow \mathbb{R}$ are defined by
\begin{align}\label{Definition functionals}
	\begin{split}
		L_1[\psi] &:= (1-\gamma_0)\left\Vert\nabla_0\psi\right\Vert^2+ \left\langle V\psi,\psi \right\rangle - \Vert \alpha_0 \vert x\vert_m^{-1} \psi\Vert^2 - \varepsilon\left\Vert\vert q(Z)\vert^{-1}_m \psi\right\Vert^2,\\
		L_2[\psi] &:= (1-\gamma_0)\left\Vert\nabla_0\psi\right\Vert^2+ \left\langle V\psi,\psi \right\rangle - \Vert \alpha_0 \vert x\vert_m^{-1} \psi\Vert^2 - \varepsilon\left\Vert\vert x\vert^{-1}_m \psi\right\Vert^2.
	\end{split}
\end{align}
We will prove that $L_1[\varphi u_{Z}]\ge 0$ and $L_2[\varphi \mathcal{V}]\ge 0$, if $\varepsilon, \gamma_0>0$ and $\kappa>0$ are sufficiently small and $R>0$ is sufficiently large. Here, $\kappa>0$ is the parameter in the definition of the cone $K(Z,\kappa)$.
\\
At first we estimate $L_1[\varphi u_{Z}]$ for an arbitrary partition $Z=(C_1,C_2)$. Note that by \eqref{H decomposed in partitions}
\begin{align}
\begin{split}
L_1[\varphi u_{Z}] &= 
		\left\langle H(Z) \varphi u_{Z}, \varphi u_{Z}\right\rangle - \gamma_0\left\Vert\nabla_{q(Z)}(\varphi u_{Z})\right\Vert^2 
		\\&\ \qquad \qquad + (1-\gamma_0)\left\Vert\nabla_{\xi(Z)}(\varphi u_{Z})\right\Vert^2	
+\langle I(Z)\varphi u_{Z}, \varphi u_{Z}\rangle 
\\&\ \qquad \qquad \qquad \qquad \qquad- \Vert \alpha_0 \vert x\vert_m^{-1} \varphi u_{Z}\Vert^2 - \varepsilon\left\Vert \vert q\left(Z\right)\vert^{-1}_m\varphi u_{Z} \right\Vert^2.
		\end{split}
\end{align}
Without loss of generality we assume that in $Z=(C_1,C_2)$ the cluster $C_1$ has two particles and $C_2$ has only one particle. Applying \eqref{ineq for ess} we get
\begin{equation}
	\langle H(Z) \varphi u_{Z}, \varphi u_{Z}\rangle  \ge \mu_0 \Vert \nabla_{q\left(Z\right)}\left(\varphi u_{Z}\right)\Vert^2
\end{equation}
for some $\mu_0>0$ independent of $\varphi$. For sufficiently small $\varepsilon>0$ and $\gamma_0 >0 $ this yields

\begin{equation}\label{no virtual level H_0(Z2) and error}
	\langle H(Z) \varphi u_{Z}, \varphi u_{Z}\rangle - \gamma_0\left\Vert\nabla_{q(Z)} (\varphi u_{Z})\right\Vert^2 -\varepsilon\left\Vert \vert q\left(Z\right)\vert_m^{-1}\varphi u_{Z}\right\Vert^2 
	\ge \frac{\mu_0}{2}\Vert \nabla_{q(Z)}\left(\varphi u_{Z}\right)\Vert^2.
\end{equation}
Therefore, we arrive at 
\begin{align}\label{1122}
\begin{split}
	L_1[\varphi u_{Z}] &\ge \frac{\mu_0}{2} \Vert \nabla_{q(Z)} (\varphi u_{Z})\Vert^2 + (1-\gamma_0)\left\Vert\nabla_{\xi(Z)} (\varphi u_{Z})\right\Vert^2\\
	&\ \qquad \qquad \ \ \ \ \ \ \ \quad +\langle I(Z) \varphi u_{Z},  \varphi u_{Z}\rangle  - \Vert \alpha_0 \vert x\vert_m^{-1} \varphi u_{Z}\Vert^2.
	\end{split}
\end{align}
On the support of $u_{Z}$ we have $\vert q(Z)\vert_m \le \kappa \vert \xi (Z)\vert_m$, which by Hardy inequality implies
\begin{equation}\label{friedrich} 
	\frac{\mu_0}{2}\Vert \nabla_{q\left(Z\right)}\left(\varphi u_{Z}\right)\Vert^2 \ge \frac{\mu_0}{8\kappa^2}\left\Vert \vert \xi(Z)\vert_m^{-1}\varphi u_{Z}\right\Vert^2.
	\end{equation}
Let $B(R)=\{x\in X_0: |x|_m \leq R\}$. Since $\mathrm{supp}\left(\varphi u_{Z}\right) \subset K\left(Z,\kappa\right)\setminus B(R)$ it holds $\vert x_{ij}\vert \ge C\vert \xi\left(Z\right)\vert_m$ for $i\in C_1,\ j\in C_2$ and some $C>0$. Therefore, by $V_{ij}\geq V_{ij}^{(1)}$ and $\vert V_{ij}^{(1)}(x_{ij})\vert\le C\vert\xi\left(Z\right)\vert_m^{-2-\nu}$ we can estimate from below the r.h.s. of \eqref{1122} as
	\begin{equation}
	\frac{\mu_0}{8\kappa^2} \Vert |\xi(Z)|_m^{-1}\varphi u_{Z} \Vert^2 - C\Vert |\xi(Z)|_m^{-1} \varphi u_{Z} \Vert^2 - \alpha_0^2\Vert |\xi(Z)|_m^{-1}\varphi u_{Z} \Vert^2 \geq 0
	\end{equation}
	for sufficiently small $\kappa>0$. Now to prove part $\textbf{(i)}$ and part $\textbf{(ii)}$ of the theorem in the case of $N=3$ it suffices to show $L_2[\mathcal{V}\varphi]\geq 0$. Note that on the support of $\mathcal{V}$ all the distances between the particles are large. Since $V_{ij}\geq V_{ij}^{(1)}$ and due to the support of $\mathcal{V}\varphi$ we have
\begin{equation}
	\vert V_{ij}^{(1)}(x_{ij})\vert \le C\vert x\vert_m^{-2-\nu}\le \varepsilon\vert x\vert_m^{-2}, \qquad i,j = 1,2,3, \ i\neq j,
\end{equation}
where $\varepsilon>0$ can be chosen arbitrarily small by choosing $R>0$ sufficiently large. This yields
\begin{equation}
	L_2[\mathcal{V}\varphi] \ge (1-\gamma_0)\Vert \nabla_0 \left(\mathcal{V}\varphi\right)\Vert^2-\left(\alpha_0^2-2\varepsilon\right)\Vert \vert x\vert_m^{-1}\varphi\mathcal{V}\Vert^2.
\end{equation}
Since $\dim X_0=6$, Hardy's inequality implies
\begin{equation}\label{3: gamma 0}
\Vert \nabla_0 (\mathcal{V}\varphi) \Vert^2\geq 4 \Vert |x|_m^{-1}\mathcal{V}\varphi\Vert^2.
\end{equation}
For $\alpha_0<2$ we can choose $0<\varepsilon<\frac{4-\alpha_0^2}{2}$ and $\gamma_0>0$ sufficiently small, such that $L_2[\varphi\mathcal{V}]\ge 0$, which completes the proof of statement \textbf{(i)} and \textbf{(ii)} for $d=3$ and $N=3$. 
\\In order to prove statement \textbf{(iii)} it suffices to note that for $\beta\in(0,2)$ and $\alpha_{ij}>0$ we have $\sum_{i,j} V_{ij}^{(2)}(x_{ij}) \geq C|x|_m^{-\beta}$. Applying statement \textbf{(iii)} of Theorem \ref{thm: Theorem 1} completes the proof for $N=3$.
\end{proof}
Now we prove the theorem for $n=3$ and $N\geq 4$.
\begin{proof}[Proof of statement \textbf{(i)} and \textbf{(ii)} of Theorem \ref{Main theorem} for $n=3$ and $N\ge 4$] 
In this part of the proof we can assume that $V_{ij}^{(2)}\equiv 0$ holds for $i,j=1,\ldots,N, \ i\not = j$.
\\Let $L[\cdot]$ be the functional defined in \eqref{Quadratic form to be estimated}. In the following we will show that $L[\varphi]\ge 0$ holds for every $0\le\alpha_0 < \frac{3N-5}{2}$ and every $\varphi \in H^1(X_0) $ with $ \mathrm{supp}(\varphi) \subset X_0\setminus B(R)$, where $B(R)=\{x\in X: |x|_m\leq R\}$ with $R>0$ sufficiently large. Analogously to the case $N=3$ we get
  \begin{equation}
 	L[\varphi] \ge \sum\limits_{Z,|Z|=2} L_1\left[\varphi u_{Z}\right]+L_2\left[\varphi \mathcal{V}_2\right],
 \end{equation}
where the functionals $L_1, L_2$ are defined in \eqref{Definition functionals} and $\mathcal{V}_2 = \sqrt{1-\sum_{Z,|Z|=2} u^2_{Z}}$. By repeating the same arguments as in the case $N=3$, one can easily show that $L_1[\varphi u_{Z}]\ge 0$ holds for all two-cluster decompositions $Z$. We only need to prove $L_2[\mathcal{V}_2\varphi]\ge 0$. 
\\Due to Theorem \ref{Lemma Intersection cones} we can find $\kappa(3)>0$, such that on the support of $\mathcal{V}_2\varphi$ the cones $K\left(Z,\kappa(3)\right)$ and $ K\left(Z',\kappa(3)\right)$ do not overlap for partitions $Z_3,Z_3'$ with $|Z|=|Z'|=3$ and $Z\neq Z'$. Applying Lemma \ref{Lemma localization error} yields
\begin{equation}
	L_2\left[\mathcal{V}_2\varphi\right]\ge \sum\limits_{Z, |Z|=3} L_1'[u_{Z}\mathcal{V}_2\varphi ] + L_2'[\mathcal{V}_3 \mathcal{V}_2\varphi],
\end{equation} 
where $\mathcal{V}_3 = \sqrt{1-\sum_{Z,|Z|=3}u_{Z}^2}$ on the support of $\mathcal{V}_2\varphi$ and 
\begin{align}
\begin{split}
		L'_1[\psi] &= \left\langle H(Z) \psi, \psi\right\rangle-\gamma_0 \left\Vert \nabla_{q(Z)}\psi\right\Vert^2 +(1-\gamma_0) \left\Vert\nabla_{\xi(Z)} \psi\right\Vert^2+\langle I(Z)\psi, \psi\rangle  \\
		&\qquad \qquad \qquad\qquad \qquad \qquad\qquad  -\left(\alpha_0^2+\varepsilon\right)\Vert \vert x\vert_m^{-1} \psi\Vert^2 - \varepsilon\left\Vert \psi \vert q\left(Z\right)\vert_m^{-1}\right\Vert^2,\\
		L'_2[\psi] &= (1-\gamma_0)\left\Vert\nabla_0 \psi\right\Vert^2+ \left\langle V\psi,\psi \right\rangle - \Vert \alpha_0 \vert x\vert_m^{-1} \psi\Vert^2 - 2\varepsilon\left\Vert \vert x\vert^{-1}_m\psi \right\Vert^2. 
\end{split}
\end{align}
Since for each cluster $C_j$ with $|C_j|>1$ in the partition $Z=(C_1,C_2,C_3)$ we have \eqref{ineq for ess}, it holds 
\begin{equation}
	\langle H(Z) \psi, \psi \rangle \ge \mu_0 \Vert \nabla_{q(Z)}\psi \Vert ^2
\end{equation}
for some $\mu_0>0$, independent of $\psi$. In addition, on the support of $u_{Z}\mathcal{V}_2$ we can estimate $\vert V_{ij}\left(x_{ij}\right)\vert \le c \vert \xi(Z)\vert_m^{-2-\nu}$ for $i,j$ belonging to different clusters in $Z$. Consequently, by the same arguments as in the estimate of $L_1[u_{Z}\varphi]$ with $|Z|=2$ we get $L_1'[u_{Z}\mathcal{V}_2\varphi]\ge 0$ with $|Z|=3$. Repeating this process, we see that to prove the theorem it suffices to show
\begin{equation}\label{remaining term N ge 4}
	L_3[\psi] := (1-\gamma_0)\left\Vert\nabla_0 \psi\right\Vert^2+ \left\langle V\psi,\psi \right\rangle - \Vert \alpha_0 \vert x\vert_m^{-1} \psi\Vert^2 - \varepsilon\left\Vert \psi \vert x\vert_m^{-1}\right\Vert^2 \ge 0
\end{equation}
for small $\varepsilon,\gamma_0>0$ and for functions $\psi\in H^1(X_0)$, which are supported outside the ball of the radius $R$ in $X_0$ in the region, where for any pair of particles $i,j$ it holds $|x_{ij}|\geq \tilde{c}|x|_m$ for some constant $\tilde{c}>0$. In this region it holds
\begin{equation}
	\vert V_{ij}\left(x_{ij}\right)\vert\le c \vert x\vert_m^{-2-\nu}.
\end{equation}
We choose $0<\varepsilon<\frac{(3(N-1)-2)^2}{4}-\alpha_0^2$, such that by Hardy's inequality in dimension $3(N-1)$ \eqref{remaining term N ge 4} holds. Now we can apply Theorem \ref{thm: Theorem 1} and conclude that zero is a simple eigenvalue of $H$ and the corresponding eigenfunction $\varphi_0$ satisfies
\begin{equation}\label{eq: decay rate}
\nabla_0\left(|x|_m^{\alpha_0}\varphi_0 \right) \in L^2(X_0) \qquad \text{and} \qquad (1+|x|_m)^{\alpha_0-1} \varphi_0 \in L^2(X_0)
\end{equation}
for every $\alpha_0 < \frac{3N-5}{2}$. This completes the proof of statement \textbf{(i)} and \textbf{(ii)} of Theorem \ref{Main theorem} in the case $n=3$ and $N\geq 4$. Finally, since Hardy's inequality holds for every $n\ge 3$, the proof of the theorem can be adapted to the case $n\ge 4$ by replacing the Hardy constant in the corresponding dimension. Statement \textbf{(iii)} of the theorem follows from statement \textbf{(iii)} of Theorem \ref{thm: Theorem 1} similar to the case of $N=3$.
\end{proof}

Theorem \ref{1: thm with symmetry} can be proved analogously to Theorem \ref{Main theorem}  by applying Theorem \ref{thm: Theorem 1'} instead of Theorem \ref{thm: Theorem 1}.

\section{Absence of the Efimov effect in $N$-particle systems with $N\ge 4$}
In this section we prove that the Efimov effect does not occur in the case of more than three particles in any dimension $n\geq3$. The main reason for this is that for such systems the virtual level is always an eigenvalue, see Theorem \ref{Main theorem}. Our proof is based on the ideas of \cite{Semjon2}, where it was shown that in case of three particles, restricted to certain symmetries, the Efimov effect does not occur as well. We will adapt this technique to arbitrary $N$-body systems.

\begin{thm}\label{No Efimov}
Let $H$ be the operator defined in \eqref{H_0[C]} with $n\geq 3$ and $N\geq 4$. Let the potentials $V_{ij}$ satisfy \eqref{2: assumptions on potential} and \eqref{2: assumptions on potential part 2}. If $n=3$ we will assume in addition that $V_{ij}\in L_{\mathrm{loc}}^2(\mathbb{R}^3)$. Further, assume that for any cluster $C$ with $|C|=N-1$ we have $H[C]\geq0$ and for any $\varepsilon\in(0,1)$
\begin{equation}\label{hhmm}
\mathcal{S}_{\mathrm{ess}}\left(-(1-\varepsilon)\Delta_0[C] + V[C]\right)=[0,\infty).
\end{equation}
Then the discrete spectrum of $H$ is finite.
\end{thm}
\begin{rem}
\begin{enumerate}
\item[\textbf{(i)}] We emphasize that in Theorem \ref{No Efimov} the cluster Hamiltonian $H[C]$ with $|C|=N-1$ may have a virtual level, which according to Theorem \ref{Main theorem} is an eigenvalue at zero. On the other hand cluster Hamiltonians $H[C']$ for clusters $C'$ with $|C'|<N-1$ are not allowed to have virtual levels due to the condition \eqref{hhmm} and the HVZ-Theorem.
\item[\textbf{(ii)}] Theorem \ref{No Efimov} can be easily generalized to the case when one of the particles has infinite mass.
\item[\textbf{(iii)}] The results of Theorem \ref{No Efimov} can be generalized to the case when the operator $H$ is considered on a subspace of functions with fixed permutational symmetry. Namely, the following theorem holds.
\end{enumerate}
\end{rem}
\begin{thm}\label{No Efimov for symmetries}
Consider the operator $H^{\sigma}=P^\sigma H$ with $n\geq3$ and $N\geq 4$, where the potentials $V_{ij}$ satisfy \eqref{2: assumptions on potential} and \eqref{2: assumptions on potential part 2}. If $n=3$, we will assume in addition that $V_{ij}\in L_{\mathrm{loc}}^2(\mathbb{R}^3)$. Let the operators $H(Z)$, the group $S(Z)$ and the inducing of the symmetry $\sigma'(Z)\prec \sigma$ be defined as in section \ref{sec}. 
\\Assume there exists $\varepsilon>0$, such that for all partitions $Z=(C_1,C_2)$ into two clusters $C_1$ and $C_2$ with $|C_1|=N-1$ or $|C_2|=N-1$ it holds
\begin{equation}
P^{\sigma'(Z)}H(Z) \geq 0 \qquad \text{and} \qquad \mathcal{S}_{\mathrm{ess}}\left( P^{\sigma'(Z)}\left( H(Z)+\varepsilon\Delta_{q(Z)} \right) \right) =[0,\infty)
\end{equation}
for all $\sigma'(Z)\prec \sigma$. Moreover, we assume that for all partitions $Z$ into two clusters $C_1$ and $C_2$ with $|C_1| \not = N-1$ and $|C_2|\not = N-1$ it holds
\begin{equation}
\mathcal{S}\left( P^{\sigma'(Z)}\left( H(Z)+\varepsilon\Delta_{q(Z)} \right) \right) =[0,\infty)
\end{equation}
for all $\sigma'(Z)\prec \sigma$. Then the discrete spectrum of $H^\sigma$ is finite. 
\end{thm}

\subsection*{Proof of Theorem \ref{No Efimov}}
Consider the functional $L_1: H^1(X_0) \rightarrow \mathbb{R}$, defined by 
\begin{equation}
L_1[\varphi]:=\langle H\varphi,\varphi\rangle-\varepsilon\Vert |x|_m^{-1} \varphi \Vert^2.
\end{equation}
Due to Lemma \ref{A3} (see in Appendix) to prove the theorem it suffices to show that there exist constants $\varepsilon>0$ and $b>0$, such that $L_1[\varphi]\geq 0$ holds for all functions $\varphi \in H^1(X_0)$ with $\supp \varphi \subset \{x\in X_0 , \ |x|_m\geq b\}$. Applying Lemma \ref{Lemma localization error} yields for partitions $Z=(C_1,C_2)$ into two clusters $C_1,C_2$
\begin{align}
	L_1[\varphi] \ge \sum\limits_{Z,|Z|=2} L_2[\varphi u_{Z}] + L_3[\varphi \mathcal{V}],
\end{align}
where $\mathcal{V}=\sqrt{1-\sum_{Z,|Z|=2} u^2_{Z}}$ and the functionals $L_2,L_3:H^1(X_0)\rightarrow \mathbb{R}$ are defined by
\begin{align}
		L_2[\psi] &:= \langle H\psi,\psi\rangle - \varepsilon \Vert \vert x\vert_m^{-1} \psi\Vert^2- \varepsilon_1 \Vert \vert q(Z)\vert_m^{-1} \psi \Vert^2_{\Omega(Z)},\\
		L_3 [\psi]&:=\langle H\psi, \psi\rangle - (\varepsilon+\varepsilon_1) \Vert \vert x\vert_m^{-1}\psi\Vert^2,
\end{align}
where 
\begin{equation}
\Omega(Z_2)\subset\{x\in X_0: |x|_m\geq b, \ \kappa'|\xi(Z)|_m\leq |q(Z)|_m \leq \kappa |\xi(Z)|_m \}.
\end{equation}
The constants $\varepsilon_1>0$ and $\kappa>0$ can be chosen arbitrarily small and $\kappa'>0$ depends on $\varepsilon_1$ and $\kappa$. At first we prove that $L_2[\varphi u_{Z}]\geq 0$. We need to distinguish between two different types of partitions $Z=(C_1,C_2)$: 
\begin{enumerate}
\item[\textbf{(i)}] $|C_1|<N-1$ and $|C_2|<N-1$,
\item[\textbf{(ii)}] $|C_1|=N-1$ or $|C_2|=N-1$.
\end{enumerate}
As it was mentioned in the remark after Theorem \ref{No Efimov} in case \textbf{(i)} the operators $H[C_1]$ and $H[C_2]$ do not have virtual levels, i.e. there exists a constant $\mu_0>0$, such that
\begin{equation}
\langle H(Z)\varphi u_{Z} ,\varphi u_{Z}\rangle \geq \mu_0 \Vert \nabla_{q(Z)} (\varphi u_{Z}) \Vert^2 
\end{equation}
holds for any $\varphi \in H^1(X_0)$. In this case analogously to the proof of Theorem \ref{Main theorem} we conclude that $L_2[\varphi  u_{Z}]\geq 0$.
\\We turn to case \textbf{(ii)}, where the cluster Hamiltonians may have virtual levels. Suppose that $|C_1|=N-1$ and that $H[C_1]$ has a virtual level. Due to Theorem \ref{Main theorem} the operator $H[C_1]$ has an eigenvalue at zero. Let $\varphi_0$ be the corresponding eigenfunction with $\Vert \varphi_0 \Vert = 1$. We estimate $L_2[\varphi u_{Z}]$ by adapting the strategy of \cite{Semjon2}. We write
\begin{equation}\label{decomposition in eigenfunction}
	\varphi u_{Z}(q(Z), \xi(Z)) = \varphi_0\big(q(Z)\big) f\big(\xi(Z\big)\big) + g\big(q(Z),\xi(Z)\big), 
\end{equation} 
where
\begin{equation}
f(\xi(Z)) = \Vert \nabla_{q(Z)}\varphi_0 \Vert^{-2} \langle \nabla_{q(Z)}(\varphi u_Z),\nabla_{q(Z)}\varphi_0 \rangle_{q(Z)}
\end{equation}
and 
\begin{equation}\label{orthogonality of phi 0 f and g}
\langle \nabla_{q(Z)} g(\cdot,\xi({Z})), \nabla_{q(Z)}\varphi_0\rangle = 0
\end{equation}
holds for almost every $\xi(Z)$. Then we have
\begin{align}\label{5.16}
\begin{split}
	L_2[\varphi u_{Z}] &= \langle H[C_1]\ g,g\rangle  +\langle H[C_1]\ \varphi_0f,\varphi_0f\rangle + 2\Re\langle H[C_1]\ g,\varphi_0f\rangle\\
	&\qquad \qquad \quad\ \ \ \ \ \ \ \ + \Vert \nabla_{\xi({Z})} (\varphi u_{Z})\Vert^2+\langle I(Z)\varphi  u_{Z},\varphi  u_{Z}\rangle\\
	&\ \qquad \qquad \ \ \ \ \ \ \ \qquad- \varepsilon \Vert \vert x\vert_m^{-1} \varphi  u_{Z}\Vert^2- \varepsilon_1 \Vert \vert q(Z)\vert_m^{-1} \varphi  u_{Z}\Vert^2_{\Omega(Z)}.
\end{split}
\end{align}
Since $H[C_1]\varphi_0=0$ the second term and the third term on the r.h.s. of \eqref{5.16} are zero. Due to the orthogonality condition \eqref{orthogonality of phi 0 f and g}, Theorem \ref{Main theorem} yields
\begin{equation}
 	 \langle H[C_1] g,g\rangle \ge \delta_0 \Vert\nabla_{q(Z)} g\Vert^2
 \end{equation}
for some $\delta_0>0$. We arrive at
\begin{align}
\begin{split}
	L_2[\varphi u_{Z}] \geq &\delta_0 \Vert\nabla_{q(Z)} g\Vert^2  + \Vert \nabla_{\xi({Z})} (\varphi u_{Z})\Vert^2+\langle I(Z)\varphi  u_{Z},\varphi  u_{Z}\rangle\\
	&\ \qquad \qquad \ \ \ \ \ \ \qquad- \varepsilon \Vert \vert x\vert_m^{-1} \varphi  u_{Z}\Vert^2- \varepsilon_1 \Vert \vert q(Z)\vert_m^{-1} \varphi  u_{Z}\Vert^2_{\Omega(Z)}.
\end{split}
\end{align}
Now since $V_{ij}\geq V_{ij}^{(1)}$, we have
\begin{align}\label{I phif small}
	\langle I(Z)\varphi u_{Z},\varphi u_{Z}\rangle &\geq\sum_{i\in C_1,j\in C_2}\langle V_{ij}^{(1)}\varphi  u_{Z},\varphi  u_{Z}\rangle\geq -\sum_{i\in C_1,j\in C_2}\langle |V_{ij}^{(1)}|\varphi  u_{Z},\varphi  u_{Z}\rangle \notag
	\\&\geq -C\Vert \vert\xi(Z)\vert_m^{-1-\frac{\nu}{2}} \varphi  u_{Z}\Vert^2	\geq -\varepsilon_2 \Vert |\nabla_{\xi(Z)} \varphi  u_{Z} \Vert^2,
\end{align}
where $\varepsilon_2>0$ can be chosen arbitrarily small by choosing $b>0$ sufficiently large. Here we used that on $\supp (\varphi  u_{Z})$ we have $\vert V_{ij}^{(1)}(x_{ij})\vert\le C \vert \xi(Z)\vert_m^{-2-\nu}\leq \frac{\varepsilon_2}{4} \vert \xi(Z)\vert_m^{-2}$ for $i,j$ belonging to different clusters. This implies
\begin{align}
\begin{split}
	L_2[\varphi u_{Z}] \geq \delta_0 \Vert\nabla_{q(Z)} g\Vert^2  +&(1-\varepsilon_2) \Vert \nabla_{\xi({Z})} (\varphi u_{Z})\Vert^2
	\\ &- \varepsilon \Vert \vert x\vert_m^{-1} \varphi  u_{Z}\Vert^2- \varepsilon_1 \Vert \vert q(Z)\vert_m^{-1} \varphi  u_{Z}\Vert^2_{\Omega(Z)}.
	\end{split}
\end{align}
Since $|x|_m^{-1} \leq |\xi(Z)|_m^{-1}$, applying Hardy's inequality yields
\begin{equation}
(1-\varepsilon_2) \Vert \nabla_{\xi(Z)} (\varphi u_{Z}) \Vert^2 - \varepsilon\Vert |x|_m^{-1}\varphi u_{Z} \Vert^2 \geq (1-\varepsilon_3)\Vert \nabla_{\xi(Z)} (\varphi u_{Z}) \Vert^2, 
\end{equation}
where $\varepsilon_3=\varepsilon_2+4\varepsilon$. This implies
\begin{equation}\label{519}
	L_2[\varphi u_{Z}] \ge \delta_0 \Vert \nabla_{q(Z)}g \Vert^2+(1-\varepsilon_3)\Vert \nabla_{\xi({Z})}(\varphi u_{Z})\Vert^2 - \varepsilon_1 \Vert \vert q(Z)\vert_m^{-1} \varphi  u_{Z}\Vert^2_{\Omega(Z)}.
\end{equation}
Let us estimate the last term on the r.h.s. of \eqref{519}. Note that
\begin{align}
	\Vert \vert q(Z)\vert_m^{-1} \varphi u_{Z} \Vert^2_{ \Omega(Z)} &\leq 2\Vert \vert q(Z)\vert_m^{-1}\varphi_0 f \Vert^2_{ \Omega(Z)}+2\Vert  \vert q(Z)\vert_m^{-1} g \Vert^2_{ \Omega(Z)}.
\end{align}
By combining the terms $\delta_0 \Vert \nabla_{q(Z)}g \Vert^2$ and $2\varepsilon_1 \Vert  \vert q(Z)\vert_m^{-1} g \Vert^2_{ \Omega(Z)}$ and applying Hardy's inequality we get for small $\varepsilon_1>0$
\begin{equation}\label{5212}
	L_2[\varphi  u_{Z}] \geq (1-\varepsilon_3)\Vert \nabla_{\xi({Z})} (\varphi u_{Z})\Vert^2 - 2\varepsilon_1 \Vert \vert q(Z)\vert_m^{-1} \varphi_0 f\Vert^2_{\Omega(Z)}.
\end{equation}
Now we estimate the last term on the r.h.s. of \eqref{5212}. Note that for $\kappa>0$ sufficiently small and $x\in \Omega(Z)$ it holds $|\xi(Z)|_m\geq \frac{b}{2}$ and
\begin{align}
\begin{split}
	\Vert \vert q(Z)\vert_m^{-1} \varphi_0f \Vert^2_{ \Omega(Z)} &\le  \int\limits_{\{\vert \xi(Z)\vert_m\ge \frac{b}{2}\}} \; \vert f\vert^2 \d \xi(Z)  \int\limits_{\tilde{\Omega}(Z,\xi(Z))} {\vert\varphi_0\vert^2}{\vert q(Z)\vert_m^{-2}} \; \d q(Z)\\
	&\le (\kappa')^{-2}\int\limits_{\{\vert \xi(Z)\vert_m\ge \frac{b}{2}\}}  \Phi{\vert f\vert^2}{\vert \xi(Z)\vert_m^{-2}}\; \d \xi(Z),
\end{split}
\end{align}
where $\tilde{\Omega}(Z,\xi(Z)) = \left\{q( Z): \ \kappa' \vert \xi(Z)\vert_m \le \vert q(Z)\vert_m\le \kappa \vert \xi(Z)\vert_m\right\}$ and 
\begin{equation}
	\Phi\left( \xi(Z)\right)=\int_{\tilde{\Omega}(Z,\xi(Z))} {\vert\varphi_0(q(Z))\vert^2} \; \d q(Z) .
\end{equation}
Since $\varphi_0$ is square-integrable in $q(Z)$, for fixed $\kappa'>0$ and any $\delta >0$ one can find $b>0$, such that $\Phi\left(\xi(Z)\right)<\delta$ holds uniformly in $\vert\xi(Z)\vert_m\ge \frac{b}{2}$. Hence, for any fixed $\kappa'>0$ and $\varepsilon_4>0$ we can choose $b>0$ sufficiently large, such that
\begin{equation}\label{inequality error q}
	\Vert \vert q(Z)\vert_m^{-1}  \varphi_0f \Vert^2_{\Omega(Z)}\le \varepsilon_4\int |\xi(Z)|_m^{-2} |f(\xi(Z))|^2 \, \mathrm{d}\xi(Z).
\end{equation}
This, together with \eqref{5212} yields
\begin{equation}\label{last}
L_2[\varphi  u_{Z}] \geq (1-\varepsilon_3)\Vert \nabla_{\xi({Z})} (\varphi u_{Z})\Vert^2 - 2\varepsilon_1 \varepsilon_4 \Vert |\xi(Z)|_m^{-1}f \Vert^2.
\end{equation}
In the following we will estimate the first term on the r.h.s. of \eqref{last}. By Hardy's inequality we have
\begin{align}\label{14141}
\Vert \nabla_{\xi({Z})} (\varphi u_{Z})\Vert^2&\geq \frac{1}{4}\Vert \varphi u_{Z} |\xi(Z)|_m^{-1}\Vert^2 = \frac{1}{4}\Vert \varphi_0f|\xi(Z)|_m^{-1}+g|\xi(Z)|_m^{-1} \Vert^2
\\ &\geq \frac{1}{4}\left(\Vert \varphi_0f|\xi(Z)|_m^{-1} \Vert^2 + \Vert g|\xi(Z)|_m^{-1} \Vert^2 - 2 \left| \langle \varphi_0 f|\xi(Z)|_m^{-1},g|\xi(Z)|_m^{-1} \rangle \right|\right).\notag
\end{align}
Note that functions $f$ and $g$ are supported in the region $|\xi(Z)|_m \geq (1+\kappa^2)^{-\frac{1}{2}}|x|_m$, where $|x|_m\geq b>0$. Hence, $f|\xi(Z)|_m^{-1} \in L^2(X_c(Z))$ and $g |\xi(Z)|_m^{-1} \in L^2(X_0)$. Under the assumptions on the potentials the domain of the operator $H(Z)$ is given by $H^2(X_0(Z))$, see \cite{reed}. Hence, due to $\varphi_0 \in H^2(X_0(Z))$ and $\langle \nabla_{q(Z)}\varphi_0, \nabla_{q(Z)}g |\xi(Z)|^{-1}\rangle = 0$, Lemma 5.3 in \cite{Semjon2} yields
\begin{equation}\label{controlled}
\left| \langle \varphi_0 f|\xi(Z)|_m^{-1},g|\xi(Z)|_m^{-1} \rangle\right| \leq 2^{-1}(1-\omega)\left(\Vert \varphi_0f|\xi(Z)|_m^{-1} \Vert^2+ \Vert g|\xi(Z)|_m^{-1} \Vert^2 \right),
\end{equation}
where $\omega>0$ depends on $\Vert \varphi_0\Vert, \Vert \nabla_{q(Z)}\varphi_0\Vert$ and $\Vert \Delta_{q(Z)} \varphi_0\Vert$ only. Combining \eqref{controlled} and \eqref{14141} we get
\begin{equation}
\Vert \nabla_{\xi(Z)} (\varphi u_{Z}) \Vert^2 \geq \frac{\omega}{2} \left(\Vert \varphi_0f|\xi(Z)|_m^{-1} \Vert^2+ \Vert g|\xi(Z)|_m^{-1} \Vert^2 \right)\geq \frac{\omega}{2} \Vert f|\xi(Z)|^{-1} \Vert^2.
\end{equation}
This, together with \eqref{last} implies $L_2[\varphi  u_{Z}]\geq 0$.
\\Thus, it remains to prove that $L_3[\varphi \mathcal{V}]\ge 0$ holds for every function $\varphi \in H^1(X_0)$ satisfying $\supp \varphi \subset \{x\in X_0, \ |x|_m\geq b\}$. For any partition $Z=(C_1,\dots, C_p)$ with $p\ge 3$ the corresponding cluster Hamiltonians $H[C_i]$ do not have virtual levels. Therefore, we can estimate the functional $L_3[\varphi \mathcal{V}]
$ in cones corresponding to partitions $Z$ into $3\le p\le N-1$ clusters, similarly to the proof of Theorem \ref{Main theorem}. In the region, which remains after the separation of cones corresponding to all $Z$ with $p\leq N-1$ it holds $\vert V_{ij}^{(1)}(x_{ij})\vert \le c\vert x\vert_m^{-2-\nu} $ for all $i\neq j$. Applying Hardy's inequality completes the proof. 
\qed	
\section{Systems of $N\geq 4$ fermions in dimension $n=1$ or $n=2$}
We consider a system of $N\ge 3$ one- or two-dimensional particles and the corresponding Hamiltonian given in \eqref{Definition Hamiltonian}, where the potentials $V_{ij}$ satisfy 
\begin{equation}\label{2: assumptions on potential in d=1,2} 
V_{ij} \in L_{\mathrm{loc}}^2(\mathbb{R}^n) \qquad \text{and} \qquad |V_{ij}(x)| \leq C|x|^{-2-\nu}, \ \text{if}\ |x|\geq A
\end{equation}
for some constants $A>0$ and $\nu>0$. Further, we assume that all particles are identical, i.e. $m_i = m_j$ for $1\le i,j\le N$ and for $i\neq j, \ k\neq l$ we have
\begin{equation}\label{2: assumption part 3 in d=1,2}
	V_{ij}(x)=V_{ij}(-x), \qquad V_{ij}(x) = V_{kl}(x), \quad x\in \R^n, \ n=1,2.
\end{equation} 
We define the space $X_0$ and the operator $H$ according to \eqref{R_0[C]} and \eqref{H_0[C]}, respectively.  Since the particles are identical, the operator $H$ is invariant under action of the group $S_N$ of permutation of particles. Let $\sigma_{\mathrm{as}}$ be the irreducible representation of $S_N$, antisymmetric with respect to permutation of each pair of particles. Let $P^{\sigma_{\mathrm{as}}}$ be the corresponding projection in $X_0$ onto the subspace of the $\sigma_{\mathrm{as}}$. We will consider the operator $H$ on the subspace $P^{\sigma_{\mathrm{as}}} L^2\left(X_0\right)$ and define $H^{\sigma_{\mathrm{as}}} = P^{\sigma_{\mathrm{as}}} H$.
Given a cluster $C$, let $S[C]$ be the subgroup of $S_N$ corresponding to permutations of particles within the cluster $C$. We denote by $\sigma_{\mathrm{as}}[C]$ the irreducible representation of $S[C]$, antisymmetric with respect to permutation of each pair of particles in $C$. Let $H^{\sigma_{\mathrm{as}}}[C] = P^{\sigma_{\mathrm{as}}[C]}H[C]$. 
\begin{df}
For an arbitrary cluster $C$ we say that the corresponding operator $H^{\sigma_{\mathrm{as}}}[C]$ has a virtual level at zero, if $H^{\sigma_{\mathrm{as}}}[C] \ge 0$ and for any $\varepsilon>0$ sufficiently small it holds
\begin{equation}
\ \ \ \ \, \mathcal{S}_\mathrm{ess}\left(P^{{\sigma_{\mathrm{as}}}[C]}\big[-(1-\varepsilon)\Delta_0[C] + V[C]\big]\right)=[0,\infty)
\end{equation}
 and
\begin{equation}
\mathcal{S}_{\mathrm{disc}}\left(P^{{\sigma_{\mathrm{as}}}[C]}\big[-(1-\varepsilon)\Delta_0[C] + V[C]\big]\right)\neq \emptyset.
\end{equation}
\end{df}
We are now ready to state the main theorem of this section.
\begin{thm}\label{111}
Consider a system of $N\geq 3$ particles in dimension $n=1$ or $n=2$. Assume that the potentials $V_{ij}$ satisfy \eqref{2: assumptions on potential in d=1,2} and \eqref{2: assumption part 3 in d=1,2}. Further, assume that $H^{\sigma_{\mathrm{as}}}$ has a virtual level at zero and for each cluster $C$ with $|C|>1$ and sufficiently small $\varepsilon>0$ it holds
	\begin{equation}
	\mathcal{S}\left(P^{{\sigma_{\mathrm{as}}}[C]}\big[-(1-\varepsilon)\Delta_0[C] + V[C]\big]\right)=[0,\infty).
	\end{equation}
Then, zero is an eigenvalue of $H^{\sigma_{\mathrm{as}}}$. 
\end{thm}
\begin{proof}
	According to Theorem \ref{thm: Theorem 1'}, it suffices to show that there exist $R>0$, $\gamma_0>0$ and $\alpha_0>1$, such that for any function $\varphi \in P^{\sigma_{\mathrm{as}}}H^1\left(X_0\right)$ with $\mathrm{supp}(\varphi) \subset X_0\setminus B(R)$ and $B(R)=\{x\in X_0: |x|_m \leq R\}$ we have
	\begin{equation}
		L[\varphi]:= \langle H^{\sigma_{\mathrm{as}}} \varphi,\varphi\rangle - \gamma_0 \Vert \nabla_0\varphi \Vert^2 - \alpha_0 \Vert \vert x\vert_m^{-1}\varphi \Vert^2\ge 0.
	\end{equation}
Note that in dimensions $n=1$ and $n=2$ Hardy's inequality holds for antisymmetric functions \cite{birman}. If $n=2$ and $N\ge 4$ or $n=1$ and $N\ge 6$ we can repeat the same arguments as in Theorem \ref{Main theorem} with $0\leq\alpha_0 < \frac{n(N-1)-2}{2}$. Hence, we only need to consider the cases $n=2, \ N=3$ and $n=1, N=3,4,5$.
\\We start with the case $n=2, \ N=3$. Since Hardy's inequality on the space of antisymmetric functions holds, by the same arguments as in the proof of Theorem \ref{Main theorem}, it suffices to show that $L_2[\varphi \mathcal{V}]\ge 0$ holds for $\varphi \in P^{\sigma_{\mathrm{as}}} H^1\left(X_0\right)$, where $L_2[\varphi\mathcal{V}]$ and $\mathcal{V}$ are defined in \eqref{Definition functionals}. Multiplication with $\mathcal{V}$ does not change the symmetry property of $\varphi$, the function $\varphi\mathcal{V}$ is antisymmetric with respect to permutations of particles. Hence, it is orthogonal to all functions depending on $\vert x\vert_m$ only. Therefore, for $\varphi\mathcal{V}$ we have (see for example in \cite{klaus}, p. 254)
	\begin{equation}\label{Hardy-type inequality}
		\Vert \nabla_0\left(\varphi\mathcal{V}\right)\Vert^2 \ge {L(L+1)}\Vert\vert x\vert_m^{-1} \varphi\mathcal{V}\Vert^2, \quad L=l+\frac{1}{2}\left(\dim X_0-3\right)
	\end{equation}
with $l=1$ and $\mathrm{dim}X_0=4$. Substituting this inequality in the formula for $L_2[\varphi\mathcal{V}]$ gives the desired estimate for $n=2$ and $N\ge 3$.\\
Now we turn to the case $n=1$ and $N=3,4,5$. Let $N=4$ or $N=5$. In this case we have $\mathrm{dim}X_0=3$ or $\mathrm{dim}X_0=4$, respectively. By the same argument as in the case of $n=2, \ N=3$, we only have to consider the functional $L_2[\varphi\mathcal{V}]$. Since $\varphi\mathcal{V}$ is orthogonal to all functions depending on $|x|_m$ only, applying the Hardy-type inequality \eqref{Hardy-type inequality} with $l=1$ and $\dim X_0=3$ or $\dim X_0=4$, respectively, yields the result for $n=1$ and $N=4, N=5$. To complete the proof it remains to consider the case $n=1$ and $N=3$. For this case we will prove the following
\begin{lem}\label{Hardy dimension 1}
Let $X_0$ be the space defined in \eqref{R_0[C]} with $n=1$ and $N=3$ and let $\psi \in C_0^1\left(X_0\right)$ be antisymmetric with respect to exchange of each pair of coordinates $(x_i,x_j)$. Then we have
\begin{equation}\label{aaa}
	\Vert \nabla_0\psi\Vert^2 \ge 9 \Vert \psi \vert x\vert_m^{-1}\Vert^2.
\end{equation}
\end{lem}
\begin{rem}
Combining the arguments of the proof of Theorem \ref{thm: Theorem 1} with the estimate \eqref{aaa} one can easily obtain an estimate on the rate of decay of virtual levels in this system. In particular, it is easy to see that a zero-energy eigenfunction $\varphi_0$ for a system of three one-dimensional fermions on the subspace of functions antisymmetric with respect to permutations of coordinates of particles satisfies $(1+|x|_m)^{2-\varepsilon}\varphi_0 \in L^2(X_0)$ for any $\varepsilon>0$.
\end{rem} \ \\
\textit{Proof of Lemma \ref{Hardy dimension 1}}.
Note that for $n=1$ and $N=3$ we have $\dim X =2$. On the plane $X_0$ we introduce the polar coordinates $\psi =\psi(\rho, \theta)$, where $\rho=\sqrt{\sum_{i=1}^3 x_i^2}$ and $\theta$ is the angle between $x$ and $\frac{1}{\sqrt 2}(1,-1,0)$. Obviously, the lines $x_1=x_2, \ x_2=x_3, \ x_1=x_3$ cut $X_0$ into six sectors, each sector having angle $\frac{\pi}{3}$. Since $\psi $ is antisymmetric with respect to reflection on these symmetry axes, we conclude that $\psi$ is a periodic function in the variable $\theta$ with period $\frac{\pi}{3}$ and $\psi (\rho,0)=0$. We represent $\psi$ as a Fourier series, i.e. we write for almost all $\rho$
 \begin{equation}\label{Fourier series}
 	\psi(\rho,\theta) = \sum\limits_{n=1}^\infty a_n(\rho) \sin (3n\theta).
 \end{equation}
Differentiating \eqref{Fourier series} we get 
\begin{equation}
	\Vert \nabla_0 \psi\Vert^2 \ge \left\Vert {\frac{1}{\rho}\frac{\partial}{\partial \theta}}\psi \right\Vert^2\ge 9 \Vert \psi \rho^{-1}\Vert^2.
\end{equation} 
This completes the proof.
\end{proof}
For the absence of the Efimov effect in systems of $N\geq 4$ one- or two-dimensional particles we get now the following result.
\begin{thm}\label{No Efimov fermions}
	Let $n=1$ or $n=2$ and consider a system of $N\ge 4$ particles. Assume that the potentials $V_{ij}$ satisfy \eqref{2: assumptions on potential in d=1,2} and \eqref{2: assumption part 3 in d=1,2}. Further, assume that for each cluster $C$ we have $H^{\sigma_{\mathrm{as}}}[C]\ge 0$ and if $1<\vert C\vert < N-1$ the operator $H^{\sigma_{\mathrm{as}}}[C]$ does not have a virtual level at zero. Then the discrete spectrum of $H^{\sigma_{\mathrm{as}}}$ is finite. 
\end{thm}
\begin{proof}
	The proof of Theorem \ref{No Efimov fermions} goes along the same line as that of Theorem \ref{No Efimov}. The only difference is that if for a cluster $C$ with $\vert C\vert=N-1$ the operator $H^{\sigma_{\mathrm{as}}}[C]$ has a virtual level, zero might be a degenerate eigenvalue of finite mulitplicity. However, in this case we can find a decompostion similar to that in \eqref{decomposition in eigenfunction} with a function $g$ which is orthogonal to the corresponding eigenspace. Repeating the arguments of the proof of Theorem \ref{No Efimov} proves Theorem \ref{No Efimov fermions}.	
\end{proof}

\begin{appendix}
\section{}
\begin{proof}[Proof of Lemma \ref{eq: cmp1}]
Let $\varepsilon>0$ and $b>0$ be fixed. Let $\tilde{b}>b$ and $u \in C^1(\mathbb{R}_+)$, such that $u(t)=1$ if $ t\leq {b}$ and $u$ is non-increasing on $[b,\infty)$. Moreover, for $t\rightarrow  {b}$ let $u'(t)\left(1-u^2(t)\right)^{-\frac{1}{2}}\rightarrow 0$. We define $v:=\sqrt{1-u^2}$,
\begin{equation}
\chi_1(x):=u\left( |x| \right) \qquad \text{and} \qquad \chi_2(x):=v\left( |x| \right).
\end{equation}
Then, since $\chi_1^2+\chi_2^2=1$ holds we have
\begin{equation}\label{eq: |x|^4}
|\nabla \chi_1|^2 +|\nabla \chi_2|^2 =  \frac{|\nabla\chi_1|^2}{\left(1-\chi_1^2 \right)} = \frac{u'(|x|)^2}{1-u(|x|)^2}.
\end{equation}
Now since $u'(|x|)\left(1-u^2(|x|)\right)^{-\frac{1}{2}}\rightarrow 0$ as $|x|\rightarrow {b}$, we can take ${b'}>b$ so close to ${b}$ that
\begin{equation}
\frac{u'(|x|)^2}{1-u(|x|)^2}\leq \varepsilon \vert x\vert^{-2},\quad \vert x\vert \in [b,b'].
\end{equation}
This together with \eqref{eq: |x|^4} implies
\begin{equation}
\left(|\nabla \chi_1|^2 +|\nabla \chi_2|^2\right) \leq \varepsilon |x|^{-2}, \quad \vert x\vert \in [b,b'].
\end{equation}
Now we define the function $u$ for $t \ge b'$ as
\begin{equation}
    u(t) = u(b') \ln\left(\frac{t}{\tilde{b}}\right)\left(\ln\left(\frac{b'}{\tilde{b}}\right)\right)^{-1}, \quad t \in [b',\tilde{b}]\quad \text{and} \quad u(t)=0, \quad t\ge \tilde{b}.
\end{equation}
Note that $u(b')$ is close to $1$, but it is strictly less than $1$. As before we set
\begin{equation}
    \chi_1(x) = u(\vert x\vert), \qquad \chi_2(x) = v(\vert x\vert), \quad \vert x\vert \ge b'.
\end{equation}
We have for $|x|\geq b'$ a.e.
\begin{equation}\label{a9}
 |\nabla \chi_1|^2+|\nabla \chi_2|^2 \leq \frac{ u^2(b')}{1-u^2(b')}\left(\ln\left( \frac{b'}{\tilde{b}} \right)\right)^{-2}|x|^{-2}.
\end{equation}
For fixed $b'$ we can choose $\tilde{b}$ so large that the r.h.s. of  \eqref{a9} can be estimated as $\varepsilon|x|^{-2}$.
\end{proof}

\begin{proof}[Proof of Lemma \ref{Lemma localization error}]
Let $\kappa>0$ and let $Z=(C_1,\ldots,C_p)$ be an arbitrary partition into $p$ clusters. For the sake of brevity we write $q$ and $\xi$ instead of $q(Z)$ and $\xi(Z)$, respectively. \\
    Let $v_1\in C^1\big(\mathbb{R}_+\big)$, such that $ v_1(t) =  1, $ if $t \ge \kappa$ and $v_1$  is non-decreasing on $[0,\kappa]$. We assume $v_1'(t)\left(1-v_1^2(t)\right)^{-\frac{1}{2}}\rightarrow 0$ as $t\rightarrow \kappa$ and define $u_1(t):=(1-v_1^2(t))^{\frac{1}{2}}$. 
\\For $0<\kappa''< \kappa$ and $x=\left(q,\xi\right)\in K(Z,\kappa)\backslash K(Z,\kappa'')$ let
    \begin{equation}
        u(x)=u_1\left( \frac{|q|_m}{|\xi|_m} \right), \qquad v(x)=v_1\left( \frac{|q|_m}{|\xi|_m} \right).
    \end{equation}
 Then we have
    \begin{equation}\label{eq: Kettenregel}
        \vert \nabla_0 u\vert^2+\vert \nabla_0 v\vert^2= \left(1-v_1^2(t)\right)^{-1}\left(v_1'\left(t\right)\right)^2\left(1+\vert q\vert_m^2\vert \xi\vert_m^{-2}\right)\vert\xi\vert_m^{-2},
    \end{equation}
    where $t=|q|_m|\xi|_m^{-1}$. Since $\kappa''\le \vert q\vert_m\vert \xi\vert_m^{-1}\le \kappa$ and $\vert x\vert_m^2 = \vert q\vert_m^2+\vert\xi\vert_m^2$ we have $|\xi|_m^{-2} \leq (1+\kappa^2)|x|_m^{-2}$. Hence, $\eqref{eq: Kettenregel}$ yields
    \begin{equation}\label{bla}
        \vert \nabla_0 v\vert^2+\vert \nabla_0 u\vert^2 \le \left(v_1'\left(t\right)\right)^2\left(1-v_1(t)^2\right)^{-1}\left(1+\kappa^2\right)^2\vert x\vert_m^{-2}.
    \end{equation}
    Since $v_1'(t)\left(1-v_1^2(t)\right)^{-\frac{1}{2}}\rightarrow 0$ as $t\rightarrow \kappa$ we can choose $\kappa''$ close to $\kappa$ to get
    \begin{equation}
    \left(v_1'\left(t\right)\right)^2\left(1-v_1(t)^2\right)^{-1}\left(1+\kappa^2\right)^2\vert x\vert_m^{-2}     < \varepsilon \vert x\vert_m^{-2}\quad \text{for} \quad x\in K(Z,\kappa)\backslash K(Z,\kappa'').
    \end{equation}
   Now we define $u$ and $v$ for $x\in K(Z,\kappa'')$. Let $0<\kappa'<\kappa''$ and set
    \begin{equation}
        v_1(t)=v_1(\kappa'') \left(\ln\left(\kappa''/ \kappa'\right)\right)^{-1} \ln(t/\kappa'), \ t \leq \kappa''.
    \end{equation}
   Let
    \begin{equation}
    v(x)=v_1\left(\frac{|q|_m}{|\xi|_m} \right),\ x\in K(Z,\kappa'')\backslash K(Z,\kappa') \ \ \text{and}\ v(x) = 0,\ x\in K(Z,\kappa').
    \end{equation}
Since $v_1(t) < v_1(\kappa'') < 1$, if $t<\kappa''$ we have
    \begin{align}\label{eq: estimate term with u}
    \begin{split}
        \left(\vert \nabla_0 u\vert^2+\vert \nabla_0 v\vert^2\right)\vert u\vert^{-2}&=\vert\nabla_0 v\vert^2\left(1-v_1^2\right)^{-1}\vert u\vert^{-2}\\
        &<\vert\nabla_0 v\vert^2(1-v_1^2(\kappa''))^{-2}
        \end{split}
    \end{align}
   and for $t=\vert q\vert_m\vert \xi\vert_m^{-1}\le \kappa''$
    \begin{equation}\label{eq: Kettenregel v}
        \vert\nabla_0 v\vert^2=\left( v_1'\left(t\right)\right)^2\left(1+\vert q\vert_m^2\vert \xi\vert_m^{-2}\right)\vert \xi\vert_m^{-2}\le \left(v_1'\left( t\right)\right)^2\left(1+(\kappa'')^2\right)\vert \xi\vert_m^{-2}.
    \end{equation}
   Note that
    \begin{equation}\label{eq: Derivative v2}
        v_1'\left( t\right) =v_1(\kappa'')\left(\ln\left(\kappa''/ \kappa'\right)\right)^{-1} t^{-1}.
    \end{equation}
Hence, combining \eqref{eq: estimate term with u}, \eqref{eq: Kettenregel v} and \eqref{eq: Derivative v2} yields
    \begin{equation}
    \left(\vert \nabla_0 u\vert^2+\vert \nabla_0 v\vert^2\right)\vert u\vert^{-2}<v_1(\kappa'')^2\left(\ln\left(\kappa''/ \kappa'\right)\right)^{-2} \left(1+(\kappa'')^2\right)t^{-2}\vert \xi\vert_m^{-2}.
    \end{equation}
Substituting $t=|q|_m|\xi|_m^{-1}$ implies
    \begin{equation}
        \left(\vert \nabla_0 u\vert^2+\vert \nabla_0 v\vert^2\right) < \varepsilon  \vert q\vert_m^{-2}\vert u\vert^{2}
    \end{equation}
   for $|q|_m<\kappa''|\xi|_m$ and $\kappa'>0$ sufficiently small. This, together with \eqref{bla} completes the proof.
\end{proof}
\begin{lem}\label{A3}
Let $h=-\Delta+V$ in $L^2(\mathbb{R}^d),\ d\geq 3,$ with $V$ satisfying \eqref{eq: A1}. Assume there exist $\varepsilon>0$ and $b>0$, such that
\begin{equation}\label{A22}
\langle h\psi,\psi\rangle - \varepsilon \langle |x|^{-2}\psi,\psi \rangle \geq 0
\end{equation}
holds for any $\psi \in H^1(\mathbb{R}^d)$ with $\supp \psi \subset \{x\in \mathbb{R}^d, \ |x|\geq b\}$. Then the following assertions hold.
\begin{enumerate}
\item[$\textbf{(i)}$] $\inf \mathcal{S}_{\mathrm{ess}}(h) \geq 0$.
\item[$\textbf{(ii)}$] The operator $h$ has at most a finite number of negative eigenvalues.
\item[$\textbf{(iii)}$] Zero is not an infinitely degenerate eigenvalue of $h$.
\item[$\textbf{(iv)}$] If the potential $V$ satisfies \eqref{1: hardyassumption} then the space $W$ of functions $\varphi \in \dot{H}^1(\mathbb{R}^d)$ with
\begin{equation}
\int_{\mathbb{R}^d} \nabla\varphi(x) \cdot \nabla\psi(x)\, \mathrm{d}x+\int_{\mathbb{R}^d} V(x)\varphi(x)\psi(x)\, \mathrm{d}x =0, \ \psi \in \dot{H}^1(\mathbb{R}^d)
\end{equation}
is at most finite-dimensional.
\end{enumerate}
\end{lem}

\begin{rem}
\begin{enumerate}
\item[$\textbf{(i)}$] The Lemma is a slightly modified variant of a part of the proof of the main Theorem in \cite{Zhislin1}. 
\item[$\textbf{(ii)}$] This result can be easily extended to the case where the operator $h$ is invariant under action of a symmetry group $G$ and we consider this operator on some symmetry space $P^\sigma L^2(\mathbb{R}^d)$, here $\sigma$ is a type of irreducible representation of $G$.
\end{enumerate}
\end{rem}
\begin{proof}
We construct a finite-dimensional subspace $M\subset L^2(\mathbb{R}^d)$, such that $\langle h\psi,\psi\rangle > 0$ holds for any $\psi \in H^1(\mathbb{R}^d)\ \left(\dot{H}^1(\mathbb{R}^d)\right)$ orthogonal to $M$. Due to Lemma \ref{eq: cmp1} we have
\begin{equation}\label{A19}
\langle h \psi,\psi \rangle \geq L[\psi \chi_1]+L[\psi \chi_2],
\end{equation}
where the functional $L$ is given by
\begin{equation}
L[\psi]= \langle h_0\psi,\psi \rangle - \varepsilon \langle |x|^{-2}\psi,\psi \rangle.
\end{equation}
Since $\psi \chi_2$ is supported outside the ball of radius $b>0$, condition \eqref{A22} implies $L[\psi\chi_2]\geq 0$. Hence, it suffices to show that $L[\psi \chi_1] >0$ holds for any $\psi \perp M$ for some finite-dimensional space $M$. By Hardy's inequality and \eqref{eq: A1} it holds
\begin{equation}
L[\psi \chi_1] \geq (1-5\varepsilon) \Vert \nabla(\chi_1 \psi)\Vert^2-C(\varepsilon)\Vert \chi_1 \psi \Vert^2.
\end{equation}  
For $k \in \mathbb{N}$ let
\begin{equation}
	M_k := \left\{\varphi_1 \chi_1, \dots, \varphi_k \chi_1\right\},
\end{equation}
where $\{\varphi_1, \dots, \varphi_k\}$ is an orthonormal set of eigenfunctions corresponding to the $k$ lowest eigenvalues of the Laplacian, acting on $L^2\left(B(b)\right)$ with Dirichlet boundary conditions, where $B(b)=\{x\in \mathbb{R}^d: |x|\leq b\}$. For $\psi \perp M_k$ we have $\psi \chi_1 \perp \varphi_1, \dots \varphi_k$, which for sufficiently large $k$ implies
\begin{equation}
	\Vert \nabla (\psi \chi_1) \Vert^2 \ge  2\left(1-\varepsilon\right)^{-1}C(\varepsilon) \Vert \psi \chi_1\Vert^2.
\end{equation}
Therefore, we conclude $L[\psi \chi_1] > 0$. This proves statements \textbf{(i)}-\textbf{(iii)}. 
\\In order to prove statement \textbf{(iv)}, we consider the operator $h_1:=h-(1+|x|)^{-3}$. The operator $h_1$ satisfies \eqref{A19} for $b>0$ sufficiently large. If the space $W$ is not finite-dimensional, then $h_1$ has an infinite number of negative eigenvalues. This is a contradiction to \textbf{(ii)}.
\end{proof}
\section{Antonets, Zhislin, Shereshevskij's Theorem}\label{Appendix B}
In the proofs of Theorem \ref{Main theorem} and Theorem \ref{No Efimov} to show that some regions of the configuration space of a $N$-particle system do not overlap we used a result, which goes back to the work \cite{cones} of M. A. Antonets, G. M. Zhislin and I. A. Shereshevskij. Since there is no English translation of this reference, below we will give the statement and the corresponding proof following the original work \cite{cones}.
\\Consider a system of $N\geq 3$ particles in dimension $d\geq 1$ with masses $m_1,\ldots ,m_N  >0$. Let $\langle \cdot,\cdot \rangle_m$ be the scalar product defined by \eqref{scalar 1}. For a cluster $C\subseteq \{1,\ldots,N\}$ let $P_0[C]$ be the projection from $X_0$ on $X_0[C]$, defined in \eqref{R0 of C}. We set
\begin{equation}\label{R_c}
X[C]=\{x=(x_1,\ldots,x_N)\in X_0 : x_i=0 \ \text{if} \ i\not \in C\},  \qquad X_c[C]=X[C] \ominus X_0[C]
\end{equation}
and denote by $P_c[C]$ the corresponding projection from $X_0$ on $X_c[C]$. We denote 
\begin{equation}
M[C] = \sum_{i\in C} m_i, \qquad M=\sum_{i=1}^N m_i \qquad \text{and} \qquad m=\min_{i=1,\ldots,N} m_i.
\end{equation}
For $P_c[C]x=(y_1,\ldots,y_N)$ holds $y_j=0$ if $j\not \in C$ and
\begin{equation}\label{center of mass subs}
y_j=x_c[C] := \frac{1}{M[C]} \sum_{i\in C} m_ix_i \qquad \text{if} \ j\in C.
\end{equation}
For a partition $Z=(C_1,\ldots,C_p)$ of order $|Z|=p$ let $X_0(Z)$ and $X_c(Z)$ be the spaces defined in \eqref{R0 of Z} and let $P_0(Z)$ and $P_c(Z)$ be the corresponding projections from $X_0$ on $X_0(Z)$ and $X_c(Z)$, respectively.
\\In the following we define the constants $\kappa'$ and $\kappa$, which will be used later as the parameters introduced in the definition of the cones $K(Z,\kappa)$ in \eqref{eq: cones}. First, we define these constants inductively as functions in $n\in \mathbb{N}$, where later $n$ will correspond to $|Z|$.
\begin{df}\label{df: kappa' and kappa}
Pick any $\kappa(1)>0$ and any $\kappa'$ with $\kappa(1)>\kappa'(1)>0$. Provided, $\kappa'(l)$ and $\kappa(l)$ are defined for some $l\ge 1$, set
\begin{equation}
	d^2(l+1)= \frac{m^3}{2M^3}\frac{\left(\kappa'(l)\right)^2}{1+\left(\kappa'(l)\right)^2}
\end{equation}
and choose $\kappa(l+1)>0$, such that the condition
\begin{equation}\label{eq: condtition kappa}
	\frac{m^3}{M^3}\frac{\left(\kappa'(l)\right)^2-\left(\kappa(l+1)\right)^2}{1+\left(\kappa'(l)\right)^2} -\left(\kappa(l+1)\right)^2 > d^2(l+1) >\left(\kappa(l+1)\right)^2\left(1+\left(\kappa(l+1)\right)^2\right)
\end{equation}
is fulfilled. Afterwards, choose $0<\kappa'(l+1)<\kappa(l+1)$.
\end{df}
\begin{rem}
It is easy to see that for fixed $\kappa'(l)$ we can always pick $\kappa(l+1)$ so small that \eqref{eq: condtition kappa} holds.
\end{rem}
\begin{thm}[Antonets, Zhislin, Shereshevskij]\label{Lemma Intersection cones}
Let $2\le l \le N-1$ and let $\kappa(l),\kappa'(l)$ be defined according to Definition \ref{df: kappa' and kappa}. Assume that $\hat{Z}, \tilde{Z}$ be two cluster decompositions of order $\vert \hat{Z}\vert = \vert \tilde{Z}\vert =l$ with $ \hat{Z} \neq \tilde{Z}$. Then
	\begin{equation}
		K\left(\hat{Z},\kappa(l)\right)\cap K\left(\tilde{Z},\kappa(l)\right) \subset\bigcup\limits_{Z:\; |Z|<l} K\left(Z,\kappa'(|Z|)\right).
	\end{equation}
\end{thm}
\begin{proof}[Proof of Theorem \ref{Lemma Intersection cones}]
In order to prove the theorem we need the following
\begin{df}
Let $\kappa'$ and $\kappa$ be chosen according to Definition \ref{df: kappa' and kappa} and let $Z$ be a partition. We define
\begin{equation}
	\mathcal{M}(Z,\kappa',\kappa)=K(Z,\kappa(|Z|))\setminus \bigcup\limits_{|Z'|<|Z|} K(Z',\kappa'(|Z'|)).
\end{equation}
\end{df}
\begin{lem}\label{Lemma B1}
For any cluster $C$, any partition $Z$ and for any $x,y\in X_0$ we have
\begin{enumerate}
\item[\textbf{(i)}] $\langle P_0[C]x,P_0[C]y\rangle_m = \frac{1}{2M[C]}\sum\limits_{i,j\in C} m_im_j \langle x_i-x_j,y_i-y_j\rangle$,
\item[\textbf{(ii)}] $\langle P_c(Z)x,P_c(Z)y\rangle_m = \frac{1}{2M}\sum\limits_{C',C''\subset Z} M[C']M[C''] \langle x_c[C']-x_c[C''], y_c[C']-y_c[C'']\rangle$.
\end{enumerate}
\end{lem}
\begin{proof}
Note that for any $x\in X_0$ we have
\begin{equation}
(P_0[C]x)_i=x_i-x_c[C] \qquad \text{if}\ i\in C
\end{equation}
and $(P_0[C]x)_i=0$ for $i\not \in C$. Hence, by definition we obtain
\begin{align}\label{eq: B1}
\langle P_0[C]&x,P_0[C]y\rangle_m \notag
\\&= \sum_{i \in C} m_i \langle x_i-x_c[C],y_i-y_c[C]\rangle \notag
\\&=\sum_{i \in C} m_i\langle x_i,y_i\rangle -\sum_{i \in C} m_i \langle x_c[C],y_i\rangle-\sum_{i \in C} m_i \langle x_i,y_c[C]\rangle  \notag + \sum_{i \in C}m_i \langle x_c[C],y_c[C] \rangle \notag
\\&= \sum_{i\in C} m_i\langle x_i,y_i\rangle - \left\langle x_c[C],\sum_{i \in C} m_iy_i\right\rangle- \left\langle \sum_{i \in C} m_i x_i,y_c[C]\right\rangle  \notag + M[C] \langle x_c[C],y_c[C] \rangle \notag
\\&=\sum_{i \in C} m_i\langle x_i,y_i\rangle -2M[C] \langle x_c[C],y_c[C] \rangle + M[C] \langle x_c[C],y_c[C] \rangle \notag
\\&= \sum_{i\in C} m_i\langle x_i,y_i\rangle-M[C]\langle x_c[C],y_c[C]\rangle.
\end{align}
On the other hand, 
\begin{align}\label{B5}
\sum_{i,j\in C} m_im_j \langle x_i-x_j,y_i-y_j\rangle &= \sum_{i\in C}m_i \sum_{j\in C}m_j\left( \langle x_i,y_i\rangle + \langle x_j,y_j\rangle -\langle x_i,y_j\rangle - \langle x_j,y_i\rangle \right) \notag
\\&=M[C]\sum_{i\in C} m_i\langle x_i,y_i\rangle +M[C]\sum_{j\in C} m_j\langle x_j,y_j\rangle \notag
\\ &\ \ \ \ \ \ \ \ \ \ \ \ -\left\langle \sum_{i\in C} m_i x_i,\sum_{j\in C} m_j y_j\right\rangle -\left\langle \sum_{j\in C} m_j x_j,\sum_{i\in C} m_i y_i\right\rangle \notag
\\&=2M[C]\sum_{i\in C} m_i\langle x_i,y_i\rangle - 2\left(M[C]\right)^2\langle x_c[C],y_c[C] \rangle \notag
\\&= 2M[C]\left( \sum_{i\in C} m_i\langle x_i,y_i\rangle - M[C]\langle x_c[C],y_c[C] \rangle \right).
\end{align}
This, together with \eqref{eq: B1} completes the proof of Lemma \ref{Lemma B1} \textbf{(i)}.
\\Now we prove \textbf{(ii)}. For $i\in C'\subset Z$ we have
\begin{equation}\label{coordinate centre of mass of cluster}
\left( P_c(Z)x \right)_i = x_c[C'],
\end{equation}
which implies
\begin{align}\label{B7}
\langle P_c(Z)x,P_c(Z)y \rangle_m &= \sum_{C'\subset Z} \sum_{j\in C'} m_j\langle x_c[C'],y_c[C']\rangle \notag
\\ &=\sum_{C'\subset Z} M[C'] \langle x_c[C'],y_c[C'] \rangle.
\end{align}
Furthermore, similar to \eqref{B5} we have
\begin{align}\label{B8}
\sum_{C',C'' \subset Z} &M[C']M[C'']\langle x_c[C']-x_c[C''] , y_c[C']-y_c[C''] \rangle\notag
\\&=\sum_{C'\subset Z}M[C']\sum_{C''\subset Z} M[C'']\Bigl( \langle x_c[C'],y_c[C']\rangle+\langle x_c[C''],y_c[C'']\rangle \notag
\\ &\ \qquad \qquad \qquad \qquad \qquad \qquad \qquad -\langle x_c[C'],y_c[C'']\rangle-\langle x_c[C''],y_c[C']\rangle \Bigr) \notag
\\&=2M\sum_{C'\subset Z}M[C']\langle x_c[C'],y_c[C']\rangle-2\left\langle \sum_{C'\subset Z}M[C']x_c[C'],\sum_{C''\subset Z}M[C'']y_c[C''] \right\rangle.
\end{align}
Since $x\in X_0$ we have $\sum_{C'\subset Z}M[C']x_c[C']=\sum_{i=1}^N m_i x_i =0$, which implies that the second term on the r.h.s. of \eqref{B8} vanishes. This yields
\begin{equation}\label{B9}
\sum_{C',C'' \subset Z}\!\!\!\! M[C']M[C'']\langle x_c[C']-x_c[C''] , y_c[C']-y_c[C'']\rangle=2M \sum_{C'\subset Z}\!\! M[C']\langle x_c[C'],y_c[C'] \rangle .
\end{equation}
Hence, combining \eqref{B9} with \eqref{B7} completes the proof of Lemma \ref{Lemma B1} \textbf{(ii)}.
\end{proof}
\begin{cor}\label{Corollary B1}
Let $Z=(C_1,\ldots,C_p)$ and $\tilde{Z}=(C_1\cup C_2,C_3,\ldots,C_p)$. Then for any $x,y\in X_0$ we have
\begin{align}
\begin{split}
&\langle P_c(Z)x,P_c(Z)y \rangle_m - \langle P_c(\tilde{Z})x, P_c(\tilde{Z})y \rangle_m
\\ &\ \ \ \ \ \ \ \ \qquad \qquad= \frac{M[C_1]M[C_2]}{M[C_1]+M[C_2]} \langle x_c[C_1]-x_c[C_2],y_c[C_1]-y_c[C_2] \rangle.
\end{split}
\end{align}
\end{cor}
\begin{proof}
The proof follows from Lemma \ref{Lemma B1} \textbf{(ii)} and
\begin{align}
x_c[C_1\cup C_2]&=\frac{1}{M[C_1\cup C_2]}\sum_{i\in C_1\cup C_2} m_ix_i \notag
\\&=\frac{1}{M[C_1]+M[C_2]}\left(\sum_{i \in C_1} m_ix_i + \sum_{j \in C_2} m_jx_j \right)\notag
 \\&= \left(M[C_1]+M[C_2]\right)^{-1} \left( M[C_1]x_c[C_1]+M[C_2]x_c[C_2]\right).
\end{align}
\end{proof}
\begin{lem}\label{lem: Intersection cones}
Let $Z=(C_1,\dots ,C_p)$ be a partition and let $C$ be a cluster with $C\not\subseteq C_i$ for any $i=1,\dots,p$. Then 
\begin{enumerate}
	\item[\textbf{(i)}] $|P_0[C]x|_m \geq d(|Z|)|P_c (Z)x|_m $ for any $x\in M(Z,\kappa',\kappa)$. 
	\item[\textbf{(ii)}] For any partition $\tilde{Z}$ containing the cluster $C$ and satisfying $|\tilde{Z}|\ge |Z|$ we have
	\begin{equation}
		K\left(\tilde{Z},\kappa(|\tilde{Z}|)\right) \cap \mathcal{M}\left(Z,\kappa'(|Z|),\kappa(|Z|)\right)=\emptyset.
	\end{equation}
\end{enumerate}
\end{lem}
\begin{proof}[Proof of Lemma \ref{lem: Intersection cones}]
Due to the assumption $C\not\subset C_l$ for all $1\leq l \leq p$ we find clusters $C_k$ and $C_n$ in the partition $Z$ with $C_k\cap C\not = \emptyset$ and $ C_n\cap C\not = \emptyset$. Let $Z'$ be the partition which is created from $Z$ by considering $C_k\cup C_n$ as a single cluster. Then $|Z'|=|Z|-1$. Let $x\in \mathcal{M}(Z,\kappa',\kappa)$, then by definition $x\in K(Z,\kappa(|Z|))$ and $x\notin K(Z',\kappa'(|Z'|))$. Therefore, we have
\begin{equation}
	\left(1+\left(\kappa'(|Z'|)\right)^2\right)|P_c(Z')x|^2_m \le |x|_m^2\le \left(1+(\kappa(|Z|))^2\right)|P_c(Z)x|_m^2.
\end{equation}
Subtracting the term $\left(1+\left(\kappa'(|Z'|)\right)^2\right)|P_c(Z')x|^2_m$ implies
\begin{align}
0&\leq \left(1+(\kappa(|Z|))^2\right)|P_c(Z)x|_m^2 - \left(1+\left(\kappa'(|Z'|)\right)^2\right)|P_c(Z')x|^2_m
\\&=\left(1+\left(\kappa'(|Z'|)\right)^2\right)\left(|P_c(Z)x|_m^2-|P_c(Z')x|_m^2\right)-\left(\left(\kappa'(|Z'|\right)^2-\left(\kappa(|Z|)\right)^2\right)|P_c(Z)x|_m^2 \notag
\end{align}
and therefore
\begin{equation}
\left(\left(\kappa'(|Z'|\right)^2-\left(\kappa(|Z|)\right)^2\right)|P_c(Z)x|_m^2\le \left(1+\left(\kappa'(|Z'|)\right)^2\right)\left(|P_c(Z)x|_m^2-|P_c(Z')x|_m^2\right).
\end{equation} 
Dividing by $\left(1+\left(\kappa'(|Z'|)\right)^2\right)$ yields
\begin{equation}\label{B20}
\frac{\left(\kappa'(|Z'|\right)^2-\left(\kappa(|Z|)\right)^2}{1+\left(\kappa'(|Z'|)\right)^2}|P_c(Z)x|_m^2 \leq |P_c(Z)x|_m^2-|P_c(Z')x|_m^2.
\end{equation}
According to Corollary \ref{Corollary B1} we have 
\begin{equation}\label{apply cor}
|P_c(Z)x|_m^2-|P_c(Z')x|_m^2 \leq \frac{M[C_k]M[C_n]}{M[C_k]+M[C_n]} \left|x_c[C_k]-x_c[C_n]\right|^2.
\end{equation}
Hence, by \eqref{B20} and \eqref{apply cor} we obtain
\begin{align}\label{eq: estimate Pc(Z)}
	\frac{\left(\kappa'(|Z'|\right)^2-\left(\kappa(|Z|)\right)^2}{1+\left(\kappa'(|Z'|)\right)^2}|P_c(Z)x|_m^2 &\le \frac{M[C_k]M[C_n]}{M[C_k]+M[C_n]} \left|x_c[C_k]-x_c[C_n]\right|^2 \notag
	\\&\le \frac{M^2}{2m}\left|x_c[C_k]-x_c[C_n]\right|^2.
\end{align}
Let us further estimate the r.h.s. of \eqref{eq: estimate Pc(Z)}. 
\\Applying Lemma \ref{Lemma B1} \textbf{(i)} for $x=y$ and $x$ replaced by $P_c(Z)x$ yields
\begin{equation}\label{eq: distance between center of mass}
\left|P_0[C]P_c(Z)x  \right|^2_m =	\frac{1}{2M[C]}\sum_{i,j\in C}m_im_j \left|\left(P_c(Z)x\right)_i-\left(P_c(Z)x\right)_j\right|^2.
	\end{equation}
Recall that $C$ has at least two particles that belong to different clusters in the partition $Z$. Hence, by choosing $s\in C\cap C_k$ and $t\in C\cap C_n$ together with \eqref{coordinate centre of mass of cluster} we get
\begin{align}
\frac{1}{2M[C]}\sum_{i,j\in C}m_im_j \left|\left(P_c(Z)x\right)_i-\left(P_c(Z)x\right)_j\right|^2 &\geq \frac{m_sm_t}{M[C]}\left|\left(P_c(Z)x\right)_s-\left(P_c(Z)x\right)_t\right|^2
\\&= \frac{m_sm_t}{M[C]}\left|x_c[C_k]-x_c[C_n]\right|^2
\\&\geq \frac{m^2}{M}\left|x_c[C_k]-x_c[C_n]\right|^2.
\end{align}
This, together with \eqref{eq: distance between center of mass} implies
\begin{equation}\label{eq: distance center of masses of clusters}
	\left|x_c[C_k]-x_c[C_n]\right|^2 \le \frac{M}{m^2}\left|P_0[C]P_c(Z)x\right|_m^2.
\end{equation}
Hence, by \eqref{eq: estimate Pc(Z)} we obtain
\begin{equation}\label{B last}
\frac{\left(\kappa'(|Z'|\right)^2-\left(\kappa(|Z|)\right)^2}{1+\left(\kappa'(|Z'|)\right)^2}|P_c(Z)x|_m^2\le \frac{M^2}{2m}\left|x_c[C_k]-x_c[C_n]\right|^2 \leq \frac{M^3}{2m^3}\left|P_0[C]P_c(Z)x\right|_m^2.
\end{equation}
Since $P_0[C]P_c(Z)=P_0[C]-P_0[C]P_0(Z)$, we can estimate
\begin{equation}
\left|P_0[C]P_c(Z)x\right|_m^2\le 2\left|P_0[C]x\right|_m^2+2\left| P_0[C]P_0(Z)x\right|_m^2\le 2\left|P_0[C]x\right|_m^2+2\left| P_0(Z)x\right|_m^2.
\end{equation}
This implies that for $x\in M(Z,\kappa',\kappa)$ we have
\begin{equation}
	\left|P_0[C]P_c(Z)x\right|_m^2\le 2 \left|P_0[C]x\right|_m^2+2\left(\kappa(|Z|)\right)^2\left| P_c(Z)x\right|_m^2.
\end{equation}
By combining this with \eqref{B last} we arrive at
\begin{equation}
\frac{m^3}{M^3}\frac{\left(\kappa'(|Z'|\right)^2-\left(\kappa(|Z|)\right)^2}{1+\left(\kappa'(|Z'|)\right)^2}|P_c(Z)x|_m^2\le   \left|P_0[C]x\right|_m^2+\left(\kappa(|Z|)\right)^2\left| P_c(Z)x\right|_m^2.
\end{equation}
Hence, applying \eqref{eq: condtition kappa} yields
\begin{equation}
\left|P_0[C]x\right|_m^2 \geq \left(\frac{m^3}{M^3}\frac{\left(\kappa'(|Z'|\right)^2-\left(\kappa(|Z|)\right)^2}{1+\left(\kappa'(|Z'|)\right)^2}- \left(\kappa(|Z|)\right)^2\right) |P_c(Z)x|_m^2 \geq \left( d(|Z|)\right)^2|P_c(Z)x|_m^2.
\end{equation}
This proves Lemma \ref{lem: Intersection cones} \textbf{(i)}. 
\\We turn to the proof of assertion \textbf{(ii)}. Assume that $K\left(\tilde{Z},\kappa(|\tilde{Z}|)\right)\cap \mathcal{M}(Z,\kappa',\kappa)\not = \emptyset$, i.e. there exists $x\in K\left(\tilde{Z},\kappa(|\tilde{Z}|)\right)\cap \mathcal{M}(Z,\kappa',\kappa)$. Since $\kappa(|Z|)\ge \kappa(|\tilde{Z}|)$, we can estimate
\begin{align}\label{eq: estimate P_0(C)}
	\left|P_0[C]x\right|_m^2&\le \left|P_0(\tilde{Z})x\right|_m^2 \le \left(\kappa(|\tilde{Z}|)\right)^2\left|P_c(\tilde{Z})x \right|_m^2\le \left(\kappa(|\tilde{Z}|)\right)^2 |x|_m^2
	\\&\le \left(\kappa(|\tilde{Z}|)\right)^2\left(1+\left(\kappa(|Z|)\right)^2\right)|P_c(Z)x|_m^2 \le \left(\kappa(|Z|)\right)^2\left(1+\left(\kappa(|Z|)\right)^2\right)|P_c(Z)x|_m^2.\notag
\end{align}
Combining inequality \eqref{eq: estimate P_0(C)} with statement \textbf{(i)} yields
\begin{equation}
	\left(d(|Z|)\right)^2\left|P_c(Z)x\right|_m^2  \le \left|P_0[C]x\right|_m^2  \le  \left(\kappa(|Z|)\right)^2\left(1+\left(\kappa(|Z|)\right)^2\right)|P_c(Z)x|_m^2.
\end{equation}
For $|P_c(Z)x|_m\not = 0$ this is a contradiction to the definition of $\kappa(|Z|)$. If $|P_c(Z)x|_m = 0$, then \eqref{eq: estimate P_0(C)} implies $x=0$, which is not possible, since by definition $0\not \in \mathcal{M}(Z,\kappa',\kappa)$. Hence, we conclude $K\left(\tilde{Z},\kappa(|\tilde{Z}|)\right)\cap \mathcal{M}(Z,\kappa',\kappa)=\emptyset$. This completes the proof of Lemma \ref{lem: Intersection cones}.
\end{proof}
Now the proof of Theorem \ref{Lemma Intersection cones} follows from Lemma \ref{lem: Intersection cones} \textbf{(ii)}. Indeed, consider two partitions $\hat{Z} =(\hat{C_1},\ldots,\hat{C_l})$ and $\tilde{Z} =(\tilde{C_1},\ldots,\tilde{C_l})$ with $\hat{Z}\not = \tilde{Z}$ and $|\hat{Z}|=|\tilde{Z}|=l$. Then there exists a cluster $\hat{C_i}$ in $\hat{Z}$ with $\hat{C_i} \not \subset \tilde{C_j}$ for all $j=1,\ldots l$. Applying Lemma \ref{lem: Intersection cones} \textbf{(ii)} completes the proof.
\end{proof}
\end{appendix}

\section*{Acknowledgements}
Simon Barth and Andreas Bitter are deeply grateful to Timo Weidl for his support and for the useful discussions. Their work was supported by the Deutsche Forschungsgemeinschaft (DFG) through the Research Training Group 1838: Spectral Theory and Quantum Systems. 
\\Semjon Vugalter gratefully acknowledges the funding by the Deutsche Forschungsgemeinschaft (DFG), German Research Foundation Project ID 258734477 - SFB 1173. Semjon Vugalter is grateful to the University of Toulon for the hospitality during his stay there.
\\ The authors thank the Mittag-Leffler Institute, where a part of the work was done during the semester program \textit{Spectral Methods in Mathematical Physics}.

\bibliography{Bib}{}
\bibliographystyle{abbrv}
\end{document}